\documentclass[runningheads]{llncs}
\usepackage[utf8]{inputenc}
\usepackage{booktabs}
\usepackage{comment}
\usepackage{tikz} 
\usepackage{pgfplots}
\pgfplotsset{compat=1.11}
\usepackage[]{algorithm2e}
\SetKwInput{KwData}{Input}
\SetKwInput{KwResult}{Output}

\usepackage{tikz}
\usepackage{interval}

\usepackage{multirow}
\usepackage{ulem}

\usepackage{url}
\usepackage{hyperref}

\usepackage{arydshln}
\usepackage{tikz,amstext}
\usetikzlibrary{positioning}
\usetikzlibrary{arrows}

\usetikzlibrary{decorations,decorations.markings}

\usepackage{amsmath,amssymb,amsfonts}
\usepackage{color} \usepackage{environ}

\newcommand{\MMfqm}{MM-$\ff{q^m}$\xspace}
\newcommand{\SMfqm}{SM-$\ff{q^m}$\xspace}
\newcommand{\SMpfqm}{SM-$\ff{q^m}^+$\xspace}
\newcommand{\MMfq}{MM-$\ff{q}$\xspace}
\newcommand{\SMfq}{SM-$\ff{q}$\xspace}

\newcommand{\mSMfqm}{\ensuremath{\text{SM-}\ff{q^m}}}
\newcommand{\mSMpfqm}{\ensuremath{\text{SM-}\ff{q^m}^+}}
\newcommand{\mMMfq}{\ensuremath{\text{MM-}\fq}}
\newcommand{\mSMfq}{\ensuremath{\text{SM-}\fq}}
\newcommand{\mSMMfq}{\ensuremath{\text{SMM-}\fq}}

\newcommand{\ff}[1]{\mathbb{F}_{#1}} \newcommand{\fq}{\ff{q}}
\newcommand{\fqm}{\ff{q^m}}

\newcommand{\Nbfqm}[1][b]{  \ensuremath{\mathcal N_{#1}^{\fqm}}}
\newcommand{\Nbfq}[1][b]{  \ensuremath{\mathcal N_{#1}^{\fq}}}
\newcommand{\Nbfqsyz}[1][b]{  \ensuremath{\mathcal N_{#1,syz}^{\fq}}}
\newcommand{\Mbfqm}[1][b]{        \ensuremath{\mathcal M_{#1}^{\ff{q^m}}}}
\newcommand{\Mbfq}[1][b]{        \ensuremath{\mathcal M_{#1}^{\ff{q}}}}

\newcommand{\pow}[2]{{#1}^{[#2]}}

\newcommand{\Fq}{\ff{q}}

\newcommand{\Fqm}{\ff{q^m}}
  \newcommand{\NN}{\mathbb{N}} \newcommand{\minor}[2]{\left| #1 \right|_{#2}}
\newcommand{\any}{*} \DeclareMathOperator{\rank}{rk}
\newcommand{\rw}[1]{\left| #1 \right|}\DeclareMathOperator{\Mat}{Mat}
\newcommand{\Matb}[1]{\Mat\left( #1\right)}

\DeclareMathOperator{\MaxMinors}{MaxMinors}
\DeclareMathOperator{\tr}{Tr}
\DeclareMathOperator{\LT}{LT}
\DeclareMathOperator{\NF}{NF}
\DeclareMathOperator{\Unfold}{Unfold}
\DeclareMathOperator{\Pos}{Pos}

\newcommand{\matRing}[3]{#1^{#2 \times #3}}
\newcommand{\mat}[1]{\boldsymbol{#1}} 

\newcommand{\zerov}{\boldsymbol{0}} 
 
\newcommand{\ident}{\mat{I}} 

\newcommand{\trsp}[1]{#1^\mathsf{T}}

\newcommand{\vsg}[2]{\langle #1 \rangle_{#2}} 
\newcommand{\cv}{\mat{c}}
\newcommand{\ev}{\mat{e}}
\newcommand{\vv}{\mat{v}}

\newcommand{\sv}{\mat{s}}

\newcommand{\uv}{\mat{u}}
\newcommand{\gv}{\mat{g}}
\newcommand{\hv}{\mat{h}}
\newcommand{\rv}{\mat{r}}

\newcommand{\xv}{\mat{x}}
\newcommand{\yv}{\mat{y}}
\newcommand{\zv}{\mat{z}}
\newcommand{\row}{\mat{r}}
 
\newcommand{\zerom}{\mat{0}}
\newcommand{\Am}{\mat{A}}
\newcommand{\Bm}{\mat{B}}

\newcommand{\Cm}{\mat{C}}
\newcommand{\Dm}{\mat{D}}
\newcommand{\Em}{\mat{E}}

\newcommand{\Gm}{\mat{G}}
\newcommand{\Hm}{\mat{H}}
\renewcommand{\Im}{\mat{I}}
\newcommand{\Lm}{\mat{L}}
\newcommand{\Mm}{\mat{M}}
\newcommand{\Pm}{\mat{P}}

\newcommand{\Rm}{\mat{R}}
\newcommand{\Sm}{\mat{S}}

\newcommand{\pri}{{^\prime}}
\newcommand{\dpr}{^{\prime\prime}}

\tikzset{
  strike through/.style={
    postaction=decorate,
    decoration={
      markings,
      mark=at position 0.5 with {
        \draw[-] (-5pt,-5pt) -- (5pt, 5pt);
      }
    }
  }
}

\everymath{\displaystyle}

\newcommand{\Iint}[2]{\lbrace #1.. #2\rbrace}

\newcommand{\Ich}{\check{I}}
\newcommand{\Jch}{\check{J}}

\newcommand{\Prob}[1]{\mathrm{Pr}\left( #1 \right) } 
\newcommand{\esp}{\mathbb{E}}
 
\newcommand{\Cc}{{\mathcal C}}

\newcommand{\cP}{{\mathcal P}}

\newcommand{\eqdef}{:=}
\newcommand{\sh}[2]{\mathbf{Sh}_{#1}\left(#2\right)}

\newcommand{\dual}[1]{#1^{\perp}}

\newcommand{\Om}[1]{\Omega \left( #1 \right)}

\usepackage[capitalize,compress]{cleveref}

 \makeatletter
\newcommand\bibalias[2]{%
  \@namedef{bibali@#1}{#2}%
}

\newtoks\biba@toks
\newcommand\acite[2][]{%
  \biba@toks{\cite#1}%
  \def\biba@comma{}%
  \def\biba@all{}%
  \@for\biba@one:=#2\do{%
    \@ifundefined{bibali@\biba@one}{%
      \edef\biba@all{\biba@all\biba@comma\biba@one}%
    }{%
      \PackageInfo{bibalias}{%
        Replacing citation `\biba@one' with `\@nameuse{bibali@\biba@one}'
      }%
      \edef\biba@all{\biba@all\biba@comma\@nameuse{bibali@\biba@one}}%
    }%
    \def\biba@comma{,}%
  }%
  \edef\biba@tmp{\the\biba@toks{\biba@all}}%
  \biba@tmp
}
\makeatother

\bibalias{Ouroboros-R}{AABBBDGHZ17}
\bibalias{RQC}{AABBBDGZ17}
\bibalias{RQC2}{AABBBDGZCH19}
\bibalias{ROLLO}{ABDGHRTZABBBO19}
\bibalias{LAKE}{ABDGHRTZ17}
\bibalias{LOCKER}{ABDGHRTZ17a}
\bibalias{RankSign}{AGHRZ17}
\bibalias{NISTR2}{AAACDKLMMPPRS20}

\newtheorem{assumption}{Assumption}
\newtheorem{modeling}{Modeling}
\newtheorem{fact}{Fact}

\crefname{modeling}{Modeling}{Modelings}

\newcommand{\repeattheorem}[1]{%
  \begingroup
  \renewcommand{\thetheorem}{\ref{#1}}%
  \expandafter\expandafter\expandafter\theorem
  \csname reptheorem@#1\endcsname
  \endtheorem
  \endgroup
}
\newcommand{\repeatproposition}[1]{%
  \begingroup
  \renewcommand{\theproposition}{\ref{#1}}%
  \expandafter\expandafter\expandafter\proposition
  \csname repproposition@#1\endcsname
  \endproposition
  \endgroup
}

\NewEnviron{reptheorem}[1]{%
  \global\expandafter\xdef\csname reptheorem@#1\endcsname{%
    \unexpanded\expandafter{\BODY}%
  }%
  \expandafter\theorem\BODY\unskip\label{#1}\endtheorem
}

\NewEnviron{repproposition}[1]{%
  \global\expandafter\xdef\csname repproposition@#1\endcsname{%
    \unexpanded\expandafter{\BODY}%
  }%
  \expandafter\proposition\BODY\unskip\label{#1}\endproposition
}

\spnewtheorem*{theorem*}{Theorem}{normalshapebfseries}{itshape}

\title{Revisiting Algebraic Attacks on MinRank and on the Rank Decoding Problem}

\author{
  Magali Bardet\inst{1,2} \and
  Pierre Briaud\inst{1,3} \and 
  Maxime Bros\inst{4} \and
  Philippe Gaborit\inst{4} \and
  Jean-Pierre Tillich\inst{1}
}
\institute{
  Inria, 2 rue Simone Iff, 75012 Paris, France\
  \and
  LITIS, University of Rouen Normandie, France\
  \and
  Sorbonne Universit\'es, UPMC Univ Paris 06 \\
  \email{pierre.briaud@inria.fr}
  \and
  Univ. Limoges, CNRS, XLIM, UMR 7252, F-87000 Limoges, France\
  \email{maxime.bros@unilim.fr}
}

\date{}

\begin{document}
\normalem
\maketitle

\begin{abstract}
The Rank Decoding problem (RD) is at the core of 
rank-based cryptography. Cryptosystems such as ROLLO and RQC, 
which made it to the second round of the NIST 
Post-Quantum Standardization Process, 
as well as the Durandal signature scheme, 
rely on it or its variants. This problem can also be seen as a structured version of MinRank, which is ubiquitous in multivariate cryptography. Recently,  \cite{BBBGNRT20,BBCGPSTV20} proposed attacks based on two new algebraic modelings, namely the MaxMinors modeling which is specific to RD and the Support-Minors modeling which applies to MinRank in general. Both  improved significantly the complexity of algebraic attacks on these two problems. In the case of RD and contrarily to what was believed up to now, these new attacks were shown to be able to outperform combinatorial attacks and this even for very small field sizes. 

However, we prove here that the analysis performed in \cite{BBCGPSTV20} for one of these attacks which consists in mixing the MaxMinors modeling with the Support-Minors modeling to solve RD is too optimistic and leads to underestimate the overall complexity. This is done by exhibiting linear dependencies between these equations and by considering an $\fqm$ version of these modelings which turns out to be instrumental for getting a better understanding of both systems. Moreover, by working over $\Fqm$ rather than over $\ff{q}$, we are able to drastically reduce the number of variables in the system and we (i) still keep enough algebraic equations to be able to solve the system, (ii) are  able to analyze rigorously the complexity of our approach. This new approach may improve the older MaxMinors approach on RD from \cite{BBBGNRT20,BBCGPSTV20} for certain parameters. We also introduce a new hybrid approach on the Support-Minors system whose impact is much more general since it applies to any MinRank problem. This technique improves significantly the complexity of the Support-Minors approach for small to moderate field sizes. 

  \keywords{Post-quantum cryptography
    \and NIST-PQC candidates
    \and rank metric code-based cryptography
    \and algebraic attack.}
\end{abstract}

\section{Introduction}

\subsubsection{Rank Metric Code-based Cryptography.}
Code-based cryptography using the rank metric, rank-based cryptography for short, 
started 30 years ago with the GPT cryptosystem \cite{GPT91} based on Gabidulin codes \cite{G85}. 
These codes can be viewed as analogues of Reed-Solomon codes in the rank metric, 
where polynomials are replaced by linearized polynomials. 
However this proposal and its variants were attacked with the Overbeck attack \cite{O05}, much in the same way 
as McEliece schemes based on Reed-Solomon codes (or variants of them) have been attacked in 
\cite{SS92,CGGOT14}.

Still, these attacks really exploited the strong algebraic 
structure of Gabidulin codes and did not rule out obtaining a 
secure version of the McEliece cryptosystem for the rank metric as we will see. One of the nice features of this metric is that it allows to exploit, 
in a much better way than the Hamming metric, codes which are linear over a very large extension field $\fqm$. Indeed, assume that we could come up with a code family which is able to decode a linear number of errors in the code length $n$ and which would remain secure when used in a McEliece scheme. In the Hamming metric, the best algorithms for solving the decoding problem for a generic linear code are exponential in this regime in $n$, whereas they are exponential in $m\cdot n$ in the case of the rank metric. This would give cryptosystems with much smaller keysize in the rank metric case, which somehow mitigates the main drawback of the original McEliece proposal that is its large keysize. This dependency of the complexity exponent in the two parameters $m$ and $n$ also allows for much finer tuning of the parameters of such schemes.

A very significant step in this direction was made with the 
Low Rank Parity Check codes (LRPC) that were  introduced in \cite{GMRZ13}. 
This type of codes made it possible to build McEliece schemes that can be viewed 
as the rank metric analogue of NTRU in the Euclidean metric \cite{HPS98} or of the MDPC cryptosystem 
in the Hamming metric \cite{MTSB12}, where the trapdoor is given by small weight vectors which 
allow efficient decoding. Contrarily to the GPT cryptosystem, this gives a cryptosystem whose security really relies on decoding an unstructured linear code and on distinguishing codes with moderate weight codewords from random linear codes. It can be argued that this second problem is similar in nature to the first one and so we have in a sense a cryptosystem whose security relies solely on the difficulty of generic decoding in the rank metric. 
This approach led to the design of several cryptosystems: \acite{GMRZ13,GRSZ14,LAKE,LOCKER}, 
and in 2019, four rank-based schemes of this form \acite{Ouroboros-R,RQC,LAKE,LOCKER} 
made it to the Second Round of the NIST Post-Quantum Standardization Process and were later merged into \acite{ROLLO,RQC2}. 

At the time of these submissions, the combinatorial attacks \cite{OJ02,GRS16,AGHT17} were thought to be the most effective against these cryptosystems, especially for small values of $q$. However, it turned out later that algebraic attacks \cite{BBBGNRT20,BBCGPSTV20} 
could be improved a great deal and may be able to outperform the combinatorial attacks. 
This is the reason why these
candidates were not kept for the Third Round, even if NIST still encourages further research on rank-based cryptography \acite{NISTR2}.  A first motivation is that these schemes still offer an interesting gain in terms of public-key size due to the algebraic structure. Another one is that the use of rank metric for wider cryptographic applications remains to be explored, and a first challenging task would already be the design of a competitive code-based signature scheme. Early attempts \acite{RankSign} based on the hash-and-sign paradigm and on structural masking where broken \cite{DT18}. More recently, a promising approach, namely Durandal,  adapting the Schnorr-Lyubashevsky framework to the rank metric, was proposed \cite{ABGHZ19}. Its security proof relies on the hardness of two problems: the first one is 
the decoding problem in the rank metric with multiple instances sharing the same support (the so-called RSL problem), while the second one is a new assumption called the Product Spaces Subspaces Indistinguishability  problem. The RSL problem was introduced in \cite{GHPT17a_sv} and also studied in \cite{BB21}. It may become instrumental to build more efficient rank-based primitives as shown by the recent work \acite{AADGZ22,BBBG22}. Finally, a third type of approach is to rely on the famous Stern's Zero-Knowledge identification protocol \cite{S93}, which is turned into a signature scheme thanks to the Fiat-Shamir transform. The advantage of this technique is that it only relies on the hardness of decoding a random linear code: first, the security is well understood, and second one can use a seed to generate the public key. This method has already inspired a long sequence of optimizations and adaptations to the rank metric setting, see for instance \acite{GSZ11,BCGMM19,BGHM20}.

\subsubsection{Rank Decoding and MinRank Problems.}
Codes used in rank metric cryptography are linear codes over an extension field $\ff{q^m}$ of degree $m$ of $\ff{q}$. An $\ff{q^m}$-linear code of length $n$ is an $\Fqm$-linear subspace of $\Fqm^n$, but its codewords can also be viewed as matrices in $\ff{q}^{m \times n}$. Indeed, if $(\beta_1,\dots,\beta_m)$ is an $\ff{q}$-basis of $\ff{q^m}$, the word $\xv = (x_1,\dots,x_n) \in \ff{q^m}^{n}$ corresponds to the matrix $\Mat(\xv) = (X_{ij})_{i,j} \in \ff{q}^{m \times n}$, where $x_j = \beta_1X_{1j} + \dots + \beta_m X_{mj}$ for $j \in \Iint{1}{n}$. The weight of $\xv$ is then defined by using the underlying rank metric on $\ff{q}^{m \times n}$, namely
$\rw{\xv} \eqdef  \rank{(\Mat{(\xv)})}$, and it is also equal to the dimension of the \emph{support} $\text{Supp}(\xv) \eqdef  \langle x_1,\dots,x_n \rangle_{\ff{q}}$. Similarly to the Hamming metric, the main source of computational hardness for rank-based cryptosystems is a decoding problem. It is the decoding problem in rank metric restricted to $\Fqm$-linear codes, namely
\begin{problem}[$(m,n,k,r)$ Rank Decoding problem (RD)]\\ 
The Rank Decoding problem of parameters $(m,n,k,r)$ is given by\\
	\indent	\emph{Input}: an $\Fqm$-linear subspace $\Cc$ of
	$\Fqm^n$, an integer $r \in \NN$, and a vector $\yv \in \Fqm^n$ such that $\rw{\yv-\cv} \leq r$ for some $\cv \in \Cc$. \\
	\indent	\emph{Output}: $\cv \in \Cc$ and an \emph{error} $\ev \in \Fqm^n$  such that $\yv=\cv+\ev$
	and $\rw{\ev} \leq r$.
\noindent We call this an $(\yv,\Cc,r)$ instance of the RD problem.
\end{problem}
\begin{remark}
{From now on, we consider that the error $\ev$ is of maximal weight $r$.  This can be done without loss of generality,  since we can run the algebraic attacks which follow for increasing values of $r' \leq r$ and since the most costly part corresponds always to $r'=r$.}
\end{remark}

Given $\sv \in \ff{q^m}^{n-k}$ and $\mat{H}\in \ff{q^m}^{(n-k) \times n}$ a parity-check matrix of an $\ff{q^m}$-linear code $\Cc$, the \emph{syndrome} version, denoted by RSD for \emph{Rank Syndrome Decoding}, asks to find $\ev \in \ff{q^m}^n$ such that $\mat{H}\trsp{\ev} = \trsp{\sv}$ and $\rw{\ev} = r$, and it is equivalent to RD. Even if RD is not known to be NP-complete, there is a randomized reduction from RD to an NP-complete problem \cite{GZ14}, namely to decoding in the Hamming metric. An RD instance can also be viewed as a structured instance of the following inhomogeneous MinRank problem. 
\begin{problem}[Inhomogeneous $(m,n,K,r)$ MinRank problem]\\
  \label{problem_minrank0}
  The  MinRank problem with parameters $(m,n,K,r)$ is given by\\
  \indent\emph{Input}: an integer $r \in \NN$ and $K+1$ 
  matrices $\mat{M}_0,\mat{M}_1,\dots,\mat{M}_K \in \matRing{\ff{q}}{m}{n}$.\\
  \indent\emph{Output}: field elements $x_1,x_2,\dots,x_K \in \ff{q}$ 
  \begin{equation*}
  \rank{\left(\mat{M}_0+\sum_{i=1}^{K}x_i\Mm_i\right)} = r.
\end{equation*}
\end{problem}

More precisely, there exists a reduction from RD to the MinRank problem \cite{FLP08}. The latter was defined and proven NP-complete in \cite{BFS99}, and it is now ubiquitous in multivariate cryptography \cite{KS99,PCYTD15,CSV17,VBCPS19,B21,TPD21,BBCPSV21,B22}. In the cryptographically relevant regime, the current best known algorithms to solve it are algebraic attacks which all have 
exponential complexity.

\subsubsection{Solving RD.} First, note that owing to the aforementioned reduction  \cite{FLP08}, all the methods for solving MinRank can be applied to the RD problem.  However, a plain MinRank solver would not be the most suitable as it forgets the $\ff{q^m}$-linear structure inherent to RD. In particular, the first attacks specific to the RD problem were of combinatorial nature \cite{CS96}. They were significantly improved in \cite{OJ02} and further refined in \cite{GRS16,AGHT18}. These works can be viewed as the continuation of the former Goubin's kernel attack on generic MinRank \cite{GC00}, which consists of first guessing sufficiently many vectors in the kernel of the rank $r$ matrix and then solving a linear system. The considerable difference in the case of RD is that the success probability of this guess can be greatly increased thanks to the $\ff{q^m}$-linearity. Another way to solve RD is provided by algebraic attacks which are not plain MinRank attacks \cite{LP06a,GRS16}. These techniques were considered to be less efficient than the combinatorial ones for a long time, especially for small values of $q$. In particular, the parameters of the rank based NIST submissions \cite{ABDGHRTZ17,ABDGHRTZ17a,AABBBDGZ17} were chosen according to the best combinatorial attacks. However, a breakthrough paper \cite{BBBGNRT20}  showed how the  $\ff{q^m}$-linear structure of the problem could be used to devise a dedicated and more efficient algebraic attack based on the so-called MaxMinors modeling. This was further improved in \cite{BBCGPSTV20}, which also introduced another algebraic modeling, the so-called Support-Minors modeling. Support-Minors is a generic MinRank modeling but it can be combined with MaxMinors in order to solve the RD problem.  In particular, this thread of work contributed to significantly break the proposed parameters for ROLLO and RQC, and these rank-based schemes have not passed the Second Round of the NIST PQC competition.

\paragraph{The MaxMinors modeling \cite{BBBGNRT20,BBCGPSTV20}.}

The attack introduced  in \cite{BBBGNRT20} relies on the following observations
\begin{itemize}
\item a vector $\uv \in \fqm^n$ is of rank $r$ iff its entries generate a subspace of $\fqm$ of dimension $r$, say 
$\vsg{s_1,\dots,s_r}{\fq}$. In such a case, there exists $\Cm \in \fq^{r \times n}$ such that
$$
\uv = (s_1,\dots,s_r) \Cm.
$$
\item
Let $(\cv,\ev)$ be the solution to RD. There exists $s_1, \dots , s_r \in \fqm$ and $\Cm \in \fq^{r \times n}$ such that
$
\yv - \cv = (s_1,\dots,s_r) \Cm,
$ because $\yv- \cv=\ev$ is of rank $\leq r$. 
If we bring in a parity check matrix $\Hm_{\yv} \in \fqm^{(n-k-1)\times n}$ of the extended code $\mathcal C + \langle \yv\rangle$ then we have
$$
(s_1,\dots,s_r) \Cm \trsp{\Hm_{\yv}}=0.
$$
This implies that the $r \times (n-k-1)$ matrix $\Cm \trsp{\Hm_{\yv}}$ is not of full rank and that all its maximal minors are equal to $0$.
By using the Cauchy-Binet formula~\eqref{eq:Cauchy-Binet}, each of these maximal minors can be expressed as a linear combination of the maximal minors 
$c_T$ of the matrix $\Cm$. Here $c_T$ denotes the maximal minor equal to the determinant of the square submatrix of $\Cm$ whose column indexes belong to $T \subset \Iint{1}{n},~\#T = r$. 
\end{itemize}
From there one readily obtains:
\begin{modeling}[MM-$\Fqm$] \label{modeling:MMfqm}~
  \begin{align}
    \MaxMinors(\Cm\trsp{\Hm_{\yv}}) =    \left\lbrace P_J\eqdef \minor{\Cm\trsp{\Hm_{\yv}}}{\any,J}: J\subset\Iint{1}{n-k-1}, \# J = r \right\rbrace \tag{\MMfqm}\label{eq:MMfqm}
  \end{align}
  {\bf Unknowns:} $\textstyle{\binom{n}{r}}$ variables $c_T \eqdef |\Cm|_{\any,T}$, $T \subset \Iint{1}{n}$, $\#T=r$, searched over $\Fq$,\\
  {\bf Equations:} $\textstyle{\binom{n-k-1}{r}}$ equations $P_J=0$, $J\subset\Iint{1}{n-k-1}$, $\#J=r$ 
  viewed as linear equations over $\Fqm$ in the $c_T$'s.
\end{modeling}
{As these polynomials have coefficients in $\fqm$ while the $c_T$'s belong to $\fq$, a standard approach is to consider a system with equations over the small field with the same solutions over $\ff{q}$. This is formalized in \cite[Notation 2]{BBCGPSTV20} with an operation\footnote{a more canonical definition will be given in Section \ref{sec:preliminaries}} which associates to a system $\mathcal F \eqdef \lbrace f_1,\dots,f_M\rbrace\subset\ff{q^m}[z_1,\dots,z_N]$ with coefficients in $\fqm$ a second system
\begin{equation*}
\label{eq:unfoldintro}
\Unfold (\mathcal{F}) := \left\lbrace f_{i,j} : 1\le i \le m, 1\le j \le M\right\rbrace\in \ff{q}[\zv]^{M \cdot m},
\end{equation*}
such that for all $j \in \{1..M\}$ and $\zv \in \ff{q}^N$, $f_j(\zv) = 0 \Leftrightarrow \left( \forall i \in \{1..m\},~f_{i,j}(\zv) = 0 \right)$, and such that the variables involved are the same. Applying this procedure to MM-$\Fqm$ yields Modeling \ref{modeling:MMfq}, denoted MM-$\Fq$, which is the relevant one for the cryptographic attack:}
\begin{modeling}[MM-$\Fq$]\label{modeling:MMfq}~
	\begin{align}
	&    \Unfold(\MaxMinors(\Cm\trsp{\Hm_{\yv}})) = 
\left\lbrace  P_{i,J} : i\in\Iint{1}{m},~J\subset\Iint{1}{n-k-1},~\#J=r\right\rbrace \tag{\MMfq}\label{eq:MMfq}
	\end{align}
	{\bf Unknowns:} $\textstyle{\binom{n}{r}}$ variables $c_T \eqdef |\Cm|_{\any,T}$, $T \subset \Iint{1}{n}$, $\#T=r$, searched over $\Fq$,\\
	{\bf Equations:} $\textstyle{m \binom{n-k-1}{r}}$ equations $P_{i,J}=0$, which are linear over $\Fq$ in the $c_T$'s.
\end{modeling}
If $\textstyle{m \binom{n-k-1}{r} \geq \binom{n}{r}-1}$, the value of the $c_T$'s may be found by solving the linear system \MMfq. This is the so-called \emph{overdetermined} case in \cite{BBCGPSTV20}. Otherwise, in the \emph{underdetermined} case, one can adopt a form of hybrid approach by adding random linear constraints on the variables to obtain another linear system that can be solved. 

\paragraph{The Support-Minors modeling \cite{BBCGPSTV20}.} 
An alternative method in the {underdetermined}  case is to rely on the Support-Minors modeling which was introduced in \cite{BBCGPSTV20}. The Support-Minors modeling is a generic MinRank modeling which is not specific to the RD problem and which can be quite effective in a certain parameter range. In particular, it turned out to be instrumental for breaking 
the third round or alternate third round multivariate finalists Rainbow and GeMSS of the NIST competition \cite{B21,B22,TPD21,BBCPSV21}. Applied to the specific RD case, Support-Minors can be explained as follows. First, rewrite $\yv - \cv = (s_1,\dots,s_r) \Cm$ in a matrix form. On the one hand, the matrix
$\Mat\left((s_1,\dots,s_r) \Cm\right)$ is readily seen to be equal to $\Sm \Cm$ where $\Sm \eqdef \Mat{(s_1,\dots,s_r)}$ and therefore we have
\begin{equation}
\label{eq:fundamental}
\Mat\left( \yv+\xv \Gm \right) = \Sm \Cm,
\end{equation}
where $\Gm$ is a generator matrix of $\Cc$, $\xv=(x_1,\dots,x_k) \in \fqm^k$ and $-\cv = \xv \Gm$. On the other hand, Equation \eqref{eq:fundamental} implies that any row 
 $\row_i$ of $\Matb{\yv+\xv \Gm}$ is in the row space of $\Cm$ and therefore all the maximal minors of the 
matrix $\textstyle{\binom{ \row_i}{ \Cm}}$ are equal to $0$. Also, it is straightforward to check that row $\row_i$ in 
$\Matb{\yv+\xv \Gm}$ is a vector whose components are affine linear forms in the $x_{i,j}$'s which are the entries of $\Matb{\xv}$. By performing Laplace expansion of any such maximal minor with respect to the first row, this minor can be written as a bilinear polynomial in the $x_{i,j}$'s on the one hand and the maximal minors $c_T$ of $\Cm$ on the other hand. This gives a bilinear system \SMfq, which as explained above, is not specific to the RD problem: we obtain a similar system for generic MinRank whose 
$x_i$ variables (coefficients in the rank $\leq r$ linear combination) play the role of our $x_{i,j}$'s. Following the terminology of \cite{BBCGPSTV20}, we call these $x_{i,j}$ variables {\em linear} variables and the $c_T$'s the {\em minor} variables.
\begin{modeling}[\SMfq]\label{modeling:SMfq} Applied to an RD instance, the SM Modeling from~\cite{BBCGPSTV20} is the system
  \begin{align}
\left\lbrace Q_{i,I}\eqdef \minor{
    \begin{pmatrix}
      \rv_i\\\Cm
    \end{pmatrix}}{\any,I} : I\subset\Iint{1}{n}, \#I=r+1, i\in\Iint{1}{m}, \rv_i=\Mat(\yv+\xv\Gm)_{i,\any}
\right\rbrace    \tag{\SMfq}\label{eq:SMfq}
  \end{align}
{\bf Unknowns:} $\textstyle{\binom{n}{r}}$ variables $c_T$ searched over $\Fq$, $k \cdot m$ variables $x_{j,j'}$ searched over $\Fq$, $j \in \Iint{1}{m}$, $j' \in \Iint{1}{k}$,\\
{\bf Equations:} $\textstyle{m \binom{n}{r+1}}$ equations $Q_{i,I}=0$ which are affine bilinear polynomials over $\Fq$ in the $x_{j,j'}$'s and in the $c_T$'s.
\end{modeling}
 If there are more linearly independent equations than bilinear monomials, the system may be solved by linearization (i.e. by replacing the monomials by single variables and then obtaining the values of these variables from solving the resulting linear system). Otherwise, the authors propose a dedicated technique to solve at higher degree by multiplying the  \SMfq equations by monomials of degree $b-1$ in the linear variables to obtain equations of degree $b$ in the linear variables and degree $1$ in the $c_T$'s . This amounts to constructing the bi-degree $(b,1)$ Macaulay matrix $\Mm_b(\mSMfq)$ whose columns are indexed by the $\mathcal{M}_b$ bi-degree $(b,1)$ monomials and then to finding a non-trivial element in the right kernel of this matrix. This approach works if the rank of $\Mm_b(\mSMfq)$ is $|\mathcal{M}_b| - 1$, so that the solution space is one-dimensional and allows to recover the original solution to the MinRank problem. The complexity of the attack is then dominated by the one of solving the system at bi-degree $(b,1)$, and for this it can be beneficial to use the Wiedemann algorithm as the Macaulay matrix is sparse enough for large values of $b$.

\subsubsection*{Solving RD by combining MaxMinors and Support-Minors.}
Recall that in the particular RD case we obtained two algebraic systems involving the same $c_T$ variables, namely the MaxMinors system \MMfq and the Support-Minors system \SMfq . This suggests to combine both modelings by multiplying the MaxMinors equations  by degree $b$ monomials in the linear variables and the Support-Minors equations by degree $b-1$ monomials in the linear variables to get equations of bi-degree $(b,1)$. In \cite{BBCGPSTV20}, it was implicitly assumed that the MaxMinors and the Support-Minors systems behave independently at higher degree, namely   
$\rank(\Mm_b(\mSMMfq)) = \rank(\Mm_b(\mMMfq)) + \rank(\Mm_b(\mSMfq))$ when this number is smaller than $\mathcal{M}_b$ which is the number of bi-degree $(b,1)$ monomials.
Here   $\Mm_b(\mMMfq)$ and $\Mm_b(\mSMMfq)$ respectively denote the Macaulay matrices of the 
MaxMinors system multiplied by the monomials of degree $b$ in the linear variables and the vertical join of $\Mm_b(\mMMfq)$ and $\Mm_b(\mSMfq)$. While it is trivial to estimate $\rank(\Mm_b(\mMMfq))$ as the MaxMinors equations \MMfq are linear in the other block of variables, note that the value obtained in \cite{BBCGPSTV20} for $\rank(\Mm_b(\mSMfq))$ is based on much more involved combinatorial arguments and remains conjectural.

\subsubsection{Contributions.} Most of our work concerns the aforementioned combined approach on the Rank Decoding Problem but some of our results will also apply to the Support-Minors strategy of \cite{BBCGPSTV20} on non-structured MinRank instances. 

First, in this combined RD approach, we show that the implicitly assumed relation $\rank(\Mm_b(\mSMMfq)) = \rank(\Mm_b(\mMMfq)) + \rank(\Mm_b(\mSMfq))$ does not hold. 
Indeed, there are linear dependencies 
between the two systems: in particular, we will prove that the
MaxMinors equations and some multiples are included in the vector space
generated by the Support-Minors equations. We will prove this by considering the ``$\fqm$-version'' of both systems. For MaxMinors, this is nothing but the original MaxMinors system \MMfqm with coefficients over $\fqm$. 
For the $\fqm$-Support-Minors modeling, this $\fqm$-version comes from a slight variation of the argument used 
in \cite{BBCGPSTV20} for obtaining the Support-Minors modeling. Instead of considering the matrix version of 
\begin{equation}
\label{eq:veryfundamental}
 \yv+\xv \Gm = (s_1,\dots,s_r) \Cm,
\end{equation}
we can directly use this equation to argue that the vector $ \yv+\xv \Gm$ is in the row space of $\Cm$, which in turn implies that all the maximal minors of the matrix
$\textstyle{\binom{\yv+\xv \Gm}{\Cm}}$ are equal to $0$. By performing Laplace expansion of these minors according to the first row, we obtain in this way $\textstyle{\binom{n}{r+1}}$ equations which are bilinear in the entries $x_i$ of $\xv$ (we still call them the {\em linear} variables) and in the maximal minors $c_T$ of $\Cm$:
\begin{modeling}[SM-$\Fqm$]\label{modeling:SMfqm}~
  \begin{align}
    \left\lbrace Q_I\eqdef \minor{
    \begin{pmatrix}
      \xv\Gm+ \yv\\\Cm
    \end{pmatrix}
    }{\any,I} : I\subset\Iint{1}{n}, \#I=r+1 \right\rbrace\tag{\SMfqm}\label{eq:SMfqm}
  \end{align}
  {\bf Unknowns:} $\textstyle{\binom{n}{r}}$ variables $c_T$ searched over $\Fq$, $k$ variables $x_1, \dots, x_k$ searched over $ \Fqm$,\\
  {\bf Equations:} $\textstyle{\binom{n}{r+1}}$ equations $Q_I=0$ for $I\subset\Iint{1}n$, $\#I=r+1$, viewed as affine bilinear equations over $\Fqm$ in the $x_i$'s on the one hand and in the $c_T$'s on the other hand.
\end{modeling}
This \SMfqm system presents the advantage of being much more compact than the original Support-Minors modeling: the number of linear variables is divided by $m$ (but the unknowns are now in $\fqm$) and the number of equations is also divided by $m$. Also, this reduced system will be very handy to study the aforementioned linear dependencies, see~\cref{sec:two}:
\begin{itemize}
\item[(i)] it is readily seen that the Support-Minors equations are the result of the Unfold operation applied to these \SMfqm equations;
\item[(ii)] it is easier to exhibit linear dependencies between the equations in \MMfqm and \SMfqm, which in turn yield 
linear dependencies between the MaxMinors and the Support-Minors equations over $\ff{q}$.
\end{itemize}
This is not the only advantage in considering \SMfqm instead of the original Support-Minors equations. It will namely be easier to understand the linear dependencies in the \SMfqm equations themselves (which also exist as we will show). Moreover, the very fact that the number of linear variables has shrunk a great deal suggests that instead of using the linearization strategy followed in \cite{BBCGPSTV20}, it might be much more favorable to
\begin{itemize}
\item[(i)] use the linear equations linking the minor variables $c_T$
  from unfolding \MMfqm (the \MMfq linear system)
  equations to substitute for some of them in \SMfqm and decrease the number of minor variables in  it to obtain a new bilinear system \SMpfqm;
\item [(ii)] multiply these equations by monomials of degree $b-1$ in the
  linear variables $x_i$ to obtain a new bi-degree $(b,1)$ system with a reduced number of bi-degree $(b,1)$ monomials and
  choose $b$ large enough so that the linearizing strategy is able to
  recover the values of these bi-degree $(b,1)$ monomials.
\end{itemize}
We call this the ``attack over $\ff{q^m}$''and we describe it in~\cref{sec:solvingfqm}, together with the count of the number of equations.
\begin{modeling}[\SMpfqm over $\ff{q^m}$]\label{modeling:final_fqm}
   \begin{align}
\mSMpfqm \eqdef \mSMfqm \mod(\mMMfq) \tag{\SMpfqm}\label{eq:SMpfqm}
   \end{align}
{\bf Unknowns:} $\textstyle{\binom{n}{r} - m\binom{n-k-1}{r}}$ variables $c_T$ searched over $ \Fq$, and $k$ unknowns $x_1,\dots,x_k$ searched over $\Fqm$,\\
{\bf Equations:} $\textstyle{\binom{n}{r+1} - \binom{n-k-1}{r+1} - (k+1) \binom{n-k-1}{r}}$ equations of the form $\widetilde{Q_I} = 0$ with 
$I \subset \Iint{1}{n}$, $\#I=r+1$, $\#(I \cap \Iint{1}{k+1} \geq 2)$, where 
$\widetilde{Q_I} =  Q_I\mod (\mMMfq)$
is the $Q_I$ equation with $c_T$ variables removed using \MMfq.
\end{modeling}

Second, we show how this ``attack over $\ff{q^m}$'' and more generally any Support-Minors based MinRank attack may benefit from a hybrid approach similar to the one presented in \cite[\S4.3]{BBCGPSTV20} on MaxMinors. There, it was used to decrease the number of minor variables. However, we will show that in our case where we consider systems with minor and linear variables, this hybrid technique has the additional benefit of decreasing the number of linear variables. Roughly speaking, our approach is to associate to a given instance of MinRank (resp. RD) $q^{a\cdot r}$ new MinRank instances (resp. RD instances) with smaller parameters for which we know that one of them has its rank $r$ matrix $\Mm$ equal to zero on a fixed set of $a \geq 0$ columns. On any of these instances and by starting from the initial modeling, we hope to find a solution of this particular form by (i) writing that $\textstyle{\binom{n}{r}-\binom{n-a}{r}}$ minors  $c_T$ should be equal to $0$, namely all those that involve one of these $a$ columns (ii) writing $a \cdot m$ linear relations between the linear variables which correspond to the $a \cdot m$ zero entries of $\Mm$. All in all, we may attack a MinRank problem of  parameters $(m,n,K,r)$ by performing $q^{a\cdot r}$ attacks on smaller instances with parameters $(m,n-a,K- a \cdot m,r)$ and such that only one of them has a solution. This is much more efficient than the straightforward hybrid approach suggested in \cite[\S 5.5]{BBCGPSTV20} which consists in fixing a few linear variables and which results only at best in a marginal gain in the complexity. Here, the gain in complexity is much more significant as shown in Subsection \ref{ss:numerical}. On a deeper level, our approach also allows to interpolate between the former combinatorial attacks \cite{GC00} and the algebraic attacks (in particular the plain Support-Minors attack).

\section{Notation and preliminaries}
\label{sec:preliminaries}
Vectors are denoted by lower case boldface letters such as
$\xv,~\ev$ and matrices by upper case letters $\Gm,~ \Mm$. The all-zero vector of length $\ell$ is denoted by $\zerov_\ell$. The $j$-th
coordinate of a vector $\xv$ is denoted by $x_{j}$ and the submatrix of a
matrix $\Mm$ formed from the rows in $I$ and columns in $J$ is denoted by
$\Mm_{I,J}$. When $I$ (resp. $J$) consists of all the rows
(resp. columns), we may use the notation $\Mm_{\any, J}$
(resp. $\Mm_{I,\any}$). Similarly, we simplify $\Mm_{i,\any} = \Mm_{\lbrace i\rbrace, \any}$ (resp. $\Mm_{\any,j} = \Mm_{\any,\lbrace j \rbrace}$) for the $i$-th row of $\Mm$ (resp. $j$-th
column of $\Mm$) and
$\Mm_{i,j}= \Mm_{\lbrace i\rbrace, \lbrace j\rbrace} $ for the entry
in row $i$ and column $j$. Finally, $\left\vert \Mm \right\vert$ is the determinant of a matrix $\Mm$, $\left\vert \Mm \right\vert_{I,J}$ is the determinant of the submatrix $\Mm_{I,J}$ and $\left\vert \Mm \right\vert_{\any,J}$ the one of $\Mm_{\any, J}$.

We will intensively use the Cauchy-Binet formula that expresses the determinant of the product of two matrices 
$A\in\mathbb K^{r\times n}$ and $B\in\mathbb K^{n\times r}$ as
\begin{equation}
\label{eq:Cauchy-Binet}
\textstyle{\minor{AB}{}=\sum_{T\subset\Iint1n, \#T=r} \minor{A}{\any,T}\minor{B}{T,\any}}.
\end{equation}

 The notation $\Iint{1}n$ stands for the set of integers from $1$ to $n$, and for any subset $J\subset\Iint{k+1}n$, we denote by $J-k$ the set $J-k = \lbrace j-k : j \in J\rbrace \subset\Iint{1}{n-k}$.

For $q$  a prime power and $m \geq 1$ an integer, let $\ff{q}$ be the finite field
with $q$ elements and let $\ff{q^m}$ be the extension of $\ff{q}$ of degree
$m$. For $x\in\ff{q^m}$ and $0\le \ell \le m-1$, we write
$\pow{x}{\ell} \eqdef x^{q^{\ell}}$ for the $\ell$-th Frobenius iterate of $x$, and this notation is extended to
matrices component by component, namely $\pow{\Mm}{\ell} \eqdef
\begin{pmatrix}
\pow{\Mm_{i,j}}{\ell}
\end{pmatrix}_{i,j}$.
We also make use of the trace operator which is the $\fq$-linear mapping from $\fqm$ to $\fq$ defined by
\[
\tr(x) \eqdef x + x^q + \dots + x^{q^{m-1}} =\sum_{\ell=0}^{m-1} \pow{x}{\ell}.
\]
In the whole paper, we consider a fixed basis $\mat\beta \eqdef (\beta_1,\dots,\beta_m)$ of $\mathbb F_{q^m}$
over $\mathbb F_q$. The {\em dual basis} $\mat\beta^{\star} \eqdef (\beta_1^\star,\dots,\beta_m^\star)$ of $\mat\beta$ is defined by
\[\tr(\beta_i\beta_j^\star) =
 \begin{cases}
   1 & \text{ if } i=j\\
   0 & \text{ otherwise}
 \end{cases}.\]
  Note that for any decomposition  in 
  $\mat{\beta}$ of the form $\textstyle{x = \sum_{i=1}^{m} x_i \beta_i} \in \ff{q^m}$ and any $i \in \Iint{1}{m}$, we can recover
 \begin{equation}\label{eq:dual}
 \tr(\beta_i^{\star} x) = x_i.
 \end{equation}
For a vector $\xv=(x_1,\dots,x_n) \in \fqm^n$ we denote by $\tr(\xv)$ the vector $(\tr(x_i))_{1 \leq i \leq n}$ where the trace is applied componentwise, and for any matrix $\Mm\in\fqm^{b\times c}$ we denote by $\tr(\Mm)=(\tr(\Mm_{i,j}))_{i,j}$. It will be helpful to notice that, thanks to the linearity of $\tr$ over $\fq$, 
\begin{align}
\label{eq:row}
  \tr(\beta_i^\star \xv) &= \Mat(\xv)_{i,\any} & \forall i \in \Iint{1}{m},\\
  \label{eq:trlinearity}
  \tr(\Cm\Mm) &= \Cm\tr(\Mm) & \text{ if } \Cm\in\fq^{a\times b}, \Mm\in\fqm^{b\times c}.
\end{align}

When looking for solutions of a polynomial system in $\ff{q}$ with coefficients in $\fqm$, it will be helpful to notice that for
$f(\zv) \in \fqm[z_1,\dots,z_N]$ and  $\xv=(x_1,\dots,x_N) \in \fq^N$, we have:
\begin{equation}
f(x_1,\dots,x_N)=0 \Longleftrightarrow \forall i \in \Iint{1}{m},~\tr(\beta_i^\star f(\xv)) =0.
\end{equation}
This motivates to define the ``unfolding'' operation which associates to an algebraic system $\mathcal F \eqdef \lbrace f_1,\dots,f_M\rbrace\subset\ff{q^m}[z_1,\dots,z_N]$ with coefficients in $\fqm$ an equivalent 
algebraic system over $\fq$ which defines the same variety over $\fq$. We call it the 
{\em associated unfolded system}:
\begin{align}
  \label{eq:unfold}
  \Unfold\left(\lbrace f_1,\dots,f_M\rbrace\right) &\eqdef \left\lbrace \tr(\beta_i^\star f_j) \mod I_q : 1\le i \le m, 1\le j \le M\right\rbrace\in \ff{q}[\zv]^{M \cdot m},
\end{align}
where we reduce the polynomials modulo the field equations, i.e. $I_q \eqdef \langle z_1^q-z_1,\dots, z_N^q-z_N\rangle$.
For one single polynomial $\textstyle{f(\zv) = \sum_{\mat \alpha\in\mathbb N^N} a_{\mat \alpha} \zv^{\mat
  \alpha}\in \ff{q^m}[\zv]}$, this reduction process reads
\begin{equation}\label{eq:trmodIq}
  \tr(\beta_i^\star f(\zv)) \mod I_q = \sum_{\mat \alpha\in\mathbb N^N} \tr(\beta_i^\star a_{\mat \alpha}) \zv^{\mat \alpha}
  \in \ff{q}[\zv].
\end{equation}
In other words, this results in applying the function $x \mapsto \tr(\beta_i^\star x)$ to each coefficient of the polynomial.

It is clear that the solutions to $\mathcal F$ in $\ff{q}^N$ are exactly the
solutions to $\Unfold(\mathcal F)$ in $\ff{q}^N$ and that any solution to
$\Unfold(\mathcal F)$ in any extension field of $\ff{q}$ is a solution to $\mathcal F$. However, note that it may be the case that $\mathcal F$ has more solutions than $\Unfold(\mathcal F)$ in some extension
field. \footnote{For instance in $\ff{q^2}$, $f=\beta_1z_1+\beta_2z_2$ admits all multiples of $(\beta_2/\beta_1,1)$ as solution, whereas $\Unfold(f) = \lbrace z_1,z_2\rbrace$ admits only $(0,0)$ as a solution in the algebraic closure of $\ff{q^2}$.}

We refer to \cite{CLO15} for basics on polynomial systems and Gröbner basis computation.
For the different results in the paper, we consider a particular
monomial ordering on our two sets of variables $x_1,\dots,x_k$ and
$c_T$'s for any subset $T$ of size $r$. The $c_T$'s are ordered with a
reverse lexicographical order according to $T$: $c_{T'}>c_{T}$ if
$t'_j=t_j$ for $j<j_0$ and $t'_{j_0}>t_{j_0}$ where
$T=\lbrace t_1<\dots<t_r\rbrace$ and
$ T'=\lbrace t'_1<\dots<t'_r\rbrace$. We then choose a grevlex (graded
reverse lexicographical) monomial ordering
$x_1>\dots>x_k>c_T$. Finally, we denote by $\LT(f)$ the leading term
of a polynomial $f$ with respect to this term order, and
$\NF(f,\mathcal G)$ the normal form of a polynomial $f$ with respect to
a system $\mathcal G$.

\section{MaxMinors and Support-Minors systems for RD instances}\label{sec:two}

In this section, we analyse the two RD modelings over $\ff{q^m}$ which take advantage of the underlying extension field structure, namely the MaxMinors (\MMfqm) and the Support-Minors (\SMfqm) systems.  

The \eqref{eq:MMfqm} system was already described in~\cite{BBBGNRT20,BBCGPSTV20} and recalled in the introduction.
The particular form of the \MMfqm polynomials $P_J$ as linear polynomials comes from the fact that these $P_J$'s can be expressed in terms of the maximal minors of $\Cm$ by using  the Cauchy-Binet formula~\eqref{eq:Cauchy-Binet}. Actually, we also use implicitly the Pl\"ucker coordinates associated to the vector space generated by the rows of $\Cm$ by defining new variables $c_T=\minor{\Cm}{\any,T}$, see~\cite[p.6]{BV88}. For $N=\textstyle{\binom{n}{r}}-1$ and $\mathbb{P}^N(\ff{q}) = \mathbb{P}(\ff{q}^{N+1})$ the projective space, the Pl\"ucker map is defined by
\begin{align*}
  p :  \lbrace \mathcal W \subset  \ff{q}^n : \dim(\mathcal W)=r\rbrace &\to \mathbb P^N(\ff{q}) & \\
  \Cm & \mapsto  (c_T)_{T\subset\Iint{1}{n}, \#T = r}
\end{align*}
where $\Cm$ is any matrix whose rows generate the vector space
$\mathcal W$. The map is well defined: any other generating matrix of $\mathcal W$ can be written $\Am\Cm$ for some invertible matrix $\Am\in GL(r,\ff{q})$, and the image $p(\Am\Cm)=\det(\Am)(c_T)_T$ is the same projective point as $(c_T)_T$. Moreover, the map is injective, and given the values of all maximal minors of a matrix it is easy to reconstruct an equivalent matrix (up to the multiplication by an invertible $\Am$) that has the same values for the minors  {(see~\cite[p.7]{BV88} for instance)}.

In our algebraic system, introducing such coordinates brings the benefit of reducing the number of solutions: for a given RD solution, there are several solutions {$\Cm \in \ff{q}^{r \times n}$} to the initial equation \eqref{eq:veryfundamental} but there are unique Plücker coordinates.
As already pointed out in \cite{BBCGPSTV20}, it is also extremely
beneficial for the computation to replace polynomials
$\minor{\Cm}{\any,T}$ with $r!$ terms of degree $r$ in the entries of $\Cm$ by single variables $c_T$'s in $\ff{q}$. Our second set of polynomials, namely the \eqref{eq:SMfqm} system, was also described in the introduction. The particular bilinear shape of these polynomials in the linear and in the minor variables follows by applying Laplace expansion along the first row $\xv \Gm +\yv$ of
$\textstyle{\binom{\xv\Gm+\yv}{\Cm}}$. Recall also that these minor
variables $c_T$ are searched over $\ff{q}$ while the linear variables $x_j$ are
searched over $\ff{q^m}$. In particular, as the \MMfqm polynomials are over $\ff{q^m}$ but linear in these $c_T$ variables, it is possible to generate $m$ times more
linear polynomials in the same variables by forming the unfolded
system \MMfq $= \Unfold$(\MMfqm) as already explained in~\cref{sec:preliminaries}. While these \MMfqm polynomials are proven to be linearly independent in \cite{BBCGPSTV20}, it is only conjectured that the resulting
\MMfq polynomials are linearly independent with overwhelming probability.

In~\cref{sec:modelingfqm}, we show that the two systems over
$\ff{q^m}$ described above are not independent: the \MMfqm polynomials are actually
included in \SMfqm; thus, adding the \MMfq polynomials to the \SMfq
system does not help to solve RD in the underdetermined case as stated in \cite{BBCGPSTV20}. Also, \SMfqm is an interesting modeling in itself to attack the RD
problem as it consists of more compact polynomials over the extension
field $\ff{q^m}$. Moreover, we are able to
formally prove the linear independence of these polynomials and more generally the exact dimension
of the vector space generated by them at each bi-degree $(b,1)$ for any $b \geq 1$, which is clearly the key quantity to evaluate the cost of such an attack. 

However, we show that it is not possible to solve the system by
using only these polynomials over $\ff{q^m}$, even at high bi-degree
$(b,1)$. Finally, note that it is also possible to
unfold the \SMfqm polynomials over $\ff{q}$ but at the cost of
multiplying the number of linear variables by a factor $m$ as we also
need to express each $\textstyle{x_j = \sum_{i=1}^m x_{i,j}\beta_i}$
in $\ff{q^m}$ as $m$ times more variables over $\ff{q}$. In~\cref{sec:modelingfq}, we show that the result of this operation is nothing more than the system \eqref{eq:SMfq} which is the Support Minors Modeling of \cite{BBCGPSTV20} applied to an RD instance, namely
\SMfq$=\Unfold$(\SMfqm). In~\cref{prop:LTQiI}, we also give a proof for the number of linearly independent polynomials in \SMfq that are not in \MMfq and which can be seen as the extra information brought by Support-Minors.

For the sake of clarity, most of the proofs are postponed
in~\cref{sec:proofs}.

\subsection{MaxMinors and Support-Minors  modelings over $\ff{q^m}$.}
\label{sec:modelingfqm}
In what follows, we always consider RD instances with a unique solution and whose rank weight is exactly $r$ instead of at most $ r$ (we may assume this, as trying 
all the weights smaller than $r$ adds at most a polynomial factor in the total complexity). Let $\Gm \in \ff{q^m}^{k \times n}$ be a full-rank generator matrix of
a linear code $\mathcal C$ of length $n$ and dimension $k$ over
$\mathbb F_{q^m}$, and let $\yv\in\mathbb F_{q^m}^n$ be the received
word affected by an error of weight $r$. With our assumption, the decoding problem amounts to
finding the unique codeword $\xv\Gm$ such that the weight of $\xv\Gm + \yv$ is
$r$. 

In this section, we analyze the link between the {\MMfqm modeling~\eqref{modeling:MMfqm}, consisting of polynomials $\textstyle{P_J=\minor{\Cm\trsp{\Hm_{\yv}}}{\any,J}}$, and the \SMfqm modeling~\eqref{modeling:SMfqm}, consisting of polynomials $Q_I=\minor{\textstyle{\binom{\xv\Gm+ \yv}{\Cm}}}{\any,I}$.} To this end, we first separate the polynomials from both systems into different sets by
defining for nonnegative integers $s$ and $i \in \Iint{1}{k}$:
\begin{align*}
 \mathcal Q_{s} &=  \lbrace Q_{I} : I \subset \Iint1n, ~\#I=r+1,~\#(I\cap \Iint1{k+1})  = s \rbrace, \\
  \mathcal Q_{\ge s} &=   \lbrace Q_I : I\subset\Iint{1}{n},~\#I=r+1,~ \#(I\cap\Iint{1}{k+1})\ge s\rbrace,\\
  \mathcal P &=  \lbrace P_J : J\subset\Iint{1}{n-k-1},~\#J=r\rbrace,\\
x_i  \mathcal P &\eqdef\lbrace x_iP: P\in\mathcal P\rbrace.  
\end{align*}
We are going to prove the following relations, where $\langle \cdot \rangle_{\ff{q}}$ means the vector space generated over $\ff{q}$:
\begin{align*}
  \mathcal Q_0 &\subset \langle \mathcal Q_1 ,\mathcal Q_{\ge 2}\rangle_{\ff{q}}   \tag{\cref{prop:Q0}}\\
  \langle \mathcal P, x_i\mathcal P : i \in \Iint{1}{k}, \mathcal Q_{\ge 2}\rangle_{\ff{q}} & =  \langle \mathcal Q_1 ,\mathcal Q_{\ge 2}\rangle_{\ff{q}}  \tag{\cref{prop:Q1}}\\
  \mathcal P, x_i\mathcal P : i \in \Iint{1}{k}, \mathcal Q_{\ge 2} & \text{ are linearly independent over } \ff{q} & \tag{ \cref{prop:LTQI}}
\end{align*}

The consequence is that if we linearize the (affine) \SMfqm system, we get
several reductions to zero and also $\textstyle\binom{n-k-1}{r}$ degree
falls\footnote{for affine systems, \emph{degree falls} correspond to linear combinations between polynomials of a given degree that yield nonzero polynomials of smaller degree.} that give the $P_J$'s polynomials. If we then eliminate $c_T$ variables using those linear polynomials, we get new reductions to zero which correspond to the $x_iP_J$'s. More generally, \cref{prop:Nb} tackles the augmented bi-degree $(b,1)$ case by giving the number of linearly
independent $Q_I$ polynomials for any $b \geq 1$ and without any particular assumption. For all these propositions, it will be helpful to notice that
\begin{fact}\label{lemma:generality}   \label{lemma:Hsystematic}
  The RD problem is equivalent to a problem where the code $\mathcal{C}$ has a generator matrix $\Gm$ in
  systematic form, i.e.
  $\Gm= \begin{pmatrix} \Im_k & * \end{pmatrix}$, where
  $\yv = \begin{pmatrix} \zerom_k & 1 & \any
  \end{pmatrix}$  and {where the extended code $\mathcal C + \langle \yv\rangle$ has a parity-check matrix $\Hm_{\yv} $ in systematic form, i.e., $\Hm_{\yv} =
  \begin{pmatrix}
    \any & \Im_{n-k-1}
  \end{pmatrix}$}.  Then,
  $\Hm \eqdef \textstyle{\binom{ \Hm_{\yv}}{\hv}}$ is a parity-check
  matrix for $\mathcal C$ for a vector $\hv =
  \begin{pmatrix}
    \any & 1 & \zerom_{n-k-1}
  \end{pmatrix}$ lying in the dual $\dual{\mathcal C}$. We have $\yv\trsp{\hv}=1$.
\end{fact}

\begin{proof}
Up to a permutation of the coordinates, we can assume
that $\Gm$ is in systematic form $\Gm= \begin{pmatrix}
  \Im_k & * \end{pmatrix}$, and up to the addition of an
element in $\mathcal C$ that $\yv=(\zerom_k \;\any)$. As $\yv$
contains an error of weight $r$, it is non-zero, so that up to a
permutation of the coordinates of the code and up to the
multiplication by a constant in $\ff{q^m}$, we assume that
$\yv$ has the given shape $\yv = \begin{pmatrix} \zerom_k & 1 & \any
\end{pmatrix}$. Now, if $\widetilde{\Gm_{\yv}} =
\begin{pmatrix}
\Im_{k+1} & \Am
\end{pmatrix}$   is a generator matrix of $\mathcal C_{\yv}$ in systematic
form, then $\Hm_{\yv} \eqdef 
\begin{pmatrix}
-\trsp{\Am} & \Im_{n-k-1}
\end{pmatrix}$ is suitable. By considering an $\hv$ linearly independent from the rows of $\Hm_{\yv}$ and such that $\yv\trsp{\hv}\ne 0$, any linear combination between $\hv$ and the rows of $\Hm_{\yv}$ still satisfies the same properties. Therefore, we may assume that $\hv =
\begin{pmatrix}
\any & \zerom_{n-k-1}
\end{pmatrix}$, and moreover we have $\yv\trsp{\hv} = h_{k+1}\ne 0$. Thus, the vector $ h_{k+1}^{-1} \hv$ is indeed of the form
$\begin{pmatrix}
\any & 1 & \zerom_{n-k-1}
\end{pmatrix}$.
\qed
\end{proof}

\begin{repproposition}{prop:Q0}
  The polynomials in $\mathcal Q_0$ can be obtained as linear
  combinations between the polynomials in $\mathcal Q_{\ge 1}$:
  \begin{align}\label{eq:Q0}
 Q_{T+k+1} &= - \sum_{Q_I\in \mathcal Q_{\ge 1} } \minor{\Hm_{\yv}}{T,I}Q_I, & \forall T\subset\Iint{1}{n-k-1},~\# T = r+1.
  \end{align}
\end{repproposition}

\begin{proof}
  This comes from the relations
$ \minor{    \begin{pmatrix}
        \xv\Gm+\yv\\\Cm
      \end{pmatrix}
    \trsp{\Hm_{\yv}}}{\any,T}=0$, see~\cref{proof:Q0} for details. \qed
\end{proof}

\begin{repproposition}{prop:LTQI}  
  The polynomials in $\mathcal P \cup \mathcal Q_{\ge 2}$ are
  linearly independent, as
  \begin{align*}
    \LT(P_J) &= c_{J+k+1} & (P_J\in\mathcal P)\\
    \LT(Q_I) &= x_{i_1}c_{I\setminus\lbrace i_1\rbrace} & (Q_I\in\mathcal Q_{\ge 2}, i_1=\min(I))
  \end{align*}
  Moreover, each variable $c_{J+k+1}$ for any
  $J\subset\Iint{1}{n-k-1},~\#J = r$ appears only as the leading term
  of $P_J$ and does not appear in any of the polynomials in
  $\mathcal Q_{\ge 2}$ nor in $P_{J'}$ with $J'\ne J$.
\end{repproposition}
\begin{proof}
  See~\cref{proof:LTQIQ1}. \qed
\end{proof}
 
\begin{repproposition}{prop:Q1}
    The polynomials in $\mathcal Q_{1}$ generate the same $\Fqm$-vector space as the polynomials
  \begin{align*}
      \mathcal  P \cup \bigcup_{j=1}^k x_j\mathcal P
  \end{align*}
  modulo the polynomials in $\mathcal Q_{\ge 2}$. More
  precisely, for any $J\subset\Iint{1}{n-k-1},~\#J =r$ and $j\in\Iint{1}{k}$ we have
  \begin{align*}
  P_J &= Q_{\lbrace k+1\rbrace \cup (J+k+1)} + \sum_{Q_I \in \mathcal Q_{\ge 2}}
       (-1)^r \minor{\Hm}{J\cup\lbrace n-k\rbrace,I}Q_I\\
      x_{j} P_J &= Q_{\lbrace {j}\rbrace \cup (J + k +1)} + \sum_{Q_I\in\mathcal Q_{\ge 2}, {j}\in I}(-1)^{1+\Pos({j},I)}\minor{\Hm_{\yv}}{J, I\setminus\lbrace {j}\rbrace}Q_I 
  \end{align*}
  where $\Pos(i_u,I)=u$ for $I=\{i_1,\dots,i_{r+1}\}$ such that $i_1 < \dots < i_{r+1}$. 
\end{repproposition}
\begin{proof}
  This comes from the relations $P_J =   (-1)^r \minor{\begin{pmatrix}
      \xv\Gm+\yv\\\Cm
    \end{pmatrix}
  \trsp{
    \begin{pmatrix}
\Hm_{\yv}\\ \hv
\end{pmatrix}
}}{\any,J\cup\lbrace n-k\rbrace}$ and $x_{j} P_J =  (-1)^r \minor{\begin{pmatrix}
      \xv\Gm+\yv\\\Cm
    \end{pmatrix}
  \trsp{
    \begin{pmatrix}
      \Hm_{\yv}\\
       \ev_j
 \end{pmatrix}
}}{\any,J\cup\lbrace n-k\rbrace}$ with $\ev_j$ the $j$-th canonical basis vector in $\ff{q}^n$, see~\cref{proof:LTQIQ1} for details. \qed
\end{proof}

To conclude this section, we have shown that the polynomials $P_J$ and
$\mathcal Q_{\ge 2}$ are linearly independent and that the polynomials
in $\mathcal Q_0$ and $\mathcal Q_1$ are redundant to the
system. Moreover, each polynomial $P_J$  can be used to eliminate the variable $c_{J+k+1}$ from
the system, so that solving $\mathcal P \cup \mathcal Q_{\ge 2}$ amounts to solve
$\mathcal Q_{\ge 2}$, that does not contain the variables
$c_{J+k+1}$. Similarly to \cite{BBCGPSTV20}, a natural approach is now to linearize at higher bi-degree $(b,1)$ after
multiplying the polynomials by linear variables. 
Here, we are able to describe precisely the $\ff{q^m}$-vector space
generated by the polynomials $\mathcal Q_{\ge 2}$ augmented at bi-degree
$(b,1)$ (see~\cref{proof:Nb} for the proof). The basis is constructed from $\mathcal Q_{\ge 2}$ without any computation:
\begin{repproposition}{prop:Nb}
For any $b\ge 1$, the $\ff{q^m}$-vector space generated by the polynomials
$\mathcal Q_{\ge 2}$ augmented at bi-degree $(b,1)$ by multiplying by monomials of degree $b-1$ in the $x_i$ variables admits the following basis:
\begin{align}
  \mathcal B_b = \left\{  {x_{i_1}}^{\alpha_{i_1}}\dots {x_{k}}^{\alpha_k}Q_I : \substack{I=\lbrace i_1<i_2<\dots<i_{r+1}\rbrace, \\i_2\le k+1, \sum_{j\ge i_1}\alpha_j=b-1}\right\}
\end{align}
In particular, it has dimension
  \begin{align}\label{eq:Nbfqm}
    \Nbfqm{} & \eqdef \sum_{i = 1}^{k} \binom{n-i}{ r}\binom{k+b-1-i}{b-1} 
    - \binom{n-k-1}{ r}\binom{k+b-1}{b},
  \end{align}
  and there are
  \begin{align}\label{eq:Mb}
    \Mbfqm{} \eqdef \binom{k+b-1}{b}\left(\binom{n}{r} - \binom{n-k-1}{r} \right)
  \end{align}
  monomials of degree $(b,1)$. We have $\Nbfqm{} < \Mbfqm{}-1$
  for any $b\ge 1$.
\end{repproposition}

As a consequence, we see that the system $Q_{\ge 2}$ always has more monomials than polynomials and cannot be solved in this way at any
degree $b$. The reason is that our initial sets of polynomials are with
coefficients in $\ff{q^m}$ and do not take into account the fact that
the $c_T$'s are searched in $\ff{q}$ (the overall system is not
zero-dimensional).
This will lead us to propose in~\cref{sec:solvingfqm} a mixed modeling by using together
polynomials over $\ff{q^m}$ and over $\ff{q}$. Prior to that, we come back to the analysis of these $\ff{q}$ polynomials in the next section.

\subsection{MaxMinors and Support-Minors  modelings over $\ff{q}$.}
\label{sec:modelingfq}
Here we consider the \emph{unfolded} systems obtained by
expressing all polynomials of  \MMfqm (resp. \SMfqm) in the fixed basis $\mat\beta \eqdef  (\beta_1,\dots,\beta_m)$ of $\ff{q^m}$ over $\ff{q}$ and
taking each component, as described in~\cref{sec:preliminaries}.  For the $P_J$'s, this unfolding process yields by definition the original \eqref{eq:MMfq} system $\lbrace P_{i,J} \rbrace_{i,J}$ \cite{BBCGPSTV20} containing $m$ times more polynomials than \MMfqm and in the same variables. For the $Q_I$'s, as the linear variables
$x_j$ lie in the extension field $\ff{q^m}$, we express each $x_j$ in the basis $\mat\beta$ as  $\textstyle x_j= \sum_{i=1}^{m}\beta_{i}x_{i,j}$, yielding $m$ times more linear variables $x_{i,j}$'s. The same unfolding technique is then applied to obtain a system $\lbrace Q_{i,I} \rbrace_{i,I}$, and \cref{prop:SMUSM} will show that it exactly corresponds to the \eqref{eq:SMfq} system defined in the introduction.

Following previous work (e.g.~\cite{BBBGNRT20,BBCGPSTV20}), we assume that the \MMfq polynomials $P_{i,J}$ are
generically as linearly independent as possible. {In other words, if the matrix $\Hm_{\yv} =
\begin{pmatrix}
\any & \Im_{n-k-1}
\end{pmatrix} \in \ff{q^m}^{(n-k-1) \times n}$ is obtained from a random code $\mathcal{C}$ of dimension $k$ and length $n$ and from a random vector $\ev \in \ff{q^m}^n$ of weight $r$ below the Gilbert-Varshamov distance, we adopt the following assumption:}
\begin{assumption}\label{ass:PiJindependent}
  The $\textstyle{m\binom{n-k-1}{r}}$ linear polynomials
  $P_{i,J}$ in the $\textstyle{\binom{n}{r}}$ variables $c_T$ generate
  an $\ff{q}$-vector space of dimension
  $\textstyle{\min\left(m\binom{n-k-1}{r},\binom{n}{r}-1 \right)}$.
\end{assumption}
{To validate this hypothesis, we have also performed experiments. The code used for these simulations can be found in \url{https://github.com/mbardet/Rank-Decoding-tools}.}
\begin{remark}
  If we consider only one $P_J$ polynomial and if we denote by $\vv_J$
  the vector of minors of $(\Hm_{\yv})_{\any,J}$ of size $r$ and by
  $\boldsymbol \mu$ the vector of minors of $\Cm$ of size $r$, then $P_J=\vv_J \trsp{\boldsymbol \mu}$
  and the rank of the system $\lbrace P_{i,J} : 1\le i \le m\rbrace$
  is exactly the rank of $\vv_J$ in rank metric. More generally, the
  number of linearly independent polynomials in $\lbrace P_{i,J} \rbrace_{i,J}$ is the co-dimension of the subfield subcode of the code
  generated by the matrix $
  \begin{pmatrix}
\vv_J
\end{pmatrix}_{J\subset\Iint{1}{n-k-1}, \#J = r}$.
\end{remark}

On the contrary, we are able to prove this result for the \SMfq
polynomials on a specific RD instance. Note that such a statement is not proven for the SM polynomials on a random MinRank instance, so in a way this $\ff{q^m}$-linear structure enables one to remove this implicit assumption of \cite{BBCGPSTV20} in the RD case.
\begin{repproposition}{prop:LTQiI}
The polynomials in $\Unfold(\mathcal Q_{\ge 2})$ satisfy
$\LT(Q_{i,I}) = x_{i,i_1}c_{I\setminus\lbrace i_1\rbrace}$ with $i_1=\min(I)\le k$.
In particular, they are all linearly independent over $\ff{q}$.
\end{repproposition}
\begin{proof}
  This comes from
  $\LT(Q_I)=x_{i_1}c_{I\setminus\lbrace
    i_1\rbrace}=\textstyle\sum_{i=1}^m \beta_{i}
  x_{i,i_1}c_{I\setminus\lbrace i_1\rbrace}$, see \cref{prop:LTQI}.
\end{proof}

Finally, we show that the polynomials from \SMfq are the unfolded polynomials obtained from \SMfqm. To this end, it may be helpful to give more details about this modeling than those given in the introduction. Let $(\yv, \Cc,r)$ an RD instance where $\Cc$ is a code of generator matrix $\Gm$ and let
\begin{eqnarray*}
\Mm_0 & \eqdef & \Mat(\yv) \\
\Mm_{\ell,j} & \eqdef & \Mat(\beta_\ell \Gm_{j,*}) \;\;\text{for $\ell \in \Iint{1}{m}$, $j \in \Iint{1}{k}$}
\end{eqnarray*}
As observed in \cite{BBCGPSTV20}, this RD problem is equivalent to a MinRank instance with rank $r$, $K=km$ and matrices 
\[(\Mm_0,\Mm_{1,1},\dots,\Mm_{i,j},\dots,\Mm_{m,k})\in\ff{q}^{m\times n}.\]
There, the Support-Minors polynomials are all the maximal minors of the matrices $\textstyle{\binom{\rv_{i}}{ \Cm}}$  for all $i \in \Iint{1}{m}$ and $\rv_i$  the $i$-th row in the solution to the 
  MinRank problem,  namely
  \begin{align*}
    \;\rv_i = \Mat\left( \yv + \sum_{\ell=1}^m \sum_{j=1}^k x_{\ell,j} \beta_\ell \Gm_{j,*} \right)_{i,*} = \tr(\beta_i^\star(\yv+\xv\Gm)),
    \end{align*}
 where the second equality follows from \eqref{eq:row}. We have
   \begin{equation*}
\rank\left(\Mm_0 + \sum_{\ell=1}^m \sum_{j=1}^k x_{\ell,j} \Mm_{\ell,j}\right) \leq r.
  \end{equation*}
We then obtain
\begin{repproposition}{prop:SMUSM}
For any $i \in \Iint{1}{m}$ and any $I \subset \Iint{1}{n},~\#I=r+1$, we have
\begin{align*}
Q_{i,I} \eqdef \minor{
            \begin{pmatrix}
\rv_i  \\\Cm
\end{pmatrix}}{\any,I}
  &= \tr\left(\beta_i^\star Q_I\right) \mod I_q,
\end{align*}
where $I_q$ is the ideal generated by all the field equations $x_{\ell,j}^q - x_{\ell,j}$ and $c_T^q - c_T$.
\end{repproposition}
\begin{proof}
The proposition basically follows from the linearity of the trace and the determinant with respect to its first row  and from \eqref{eq:trmodIq}:
 \begin{align*}
   \tr\left( \beta_i^\star \minor{
             \begin{pmatrix}
 \yv +  \xv \Gm  \\\Cm
             \end{pmatrix}}{\any,I} \right) \mod I_q 
            &= & 
            \tr\left( \minor{
             \begin{pmatrix}
 \beta_i^\star \left(\yv +  \xv\Gm\right)  \\\Cm
            \end{pmatrix}}{\any,I} \right) \mod I_q\\
&  =   &  \minor{
            \begin{pmatrix}
 \tr\left(   \beta_i^\star\left( \yv +  \xv \Gm\right) \right) \\\Cm
            \end{pmatrix}}{\any,I}    
=           \minor{
            \begin{pmatrix}
\rv_i  \\\Cm
            \end{pmatrix}}{\any,I}. 
 \end{align*}
\qed
\end{proof}
\section{Algebraic approach to solve RD by combining  \SMfqm and \MMfq}\label{sec:solvingfqm}
From the material presented in the previous section, we conclude that the polynomials $P_{i,J}$  over $\ff{q}$ (i.e. \MMfq) are necessary to solve
the system: without them we cannot solve RD since the previously considered ideal without these polynomials was not zero-dimensional. However, we also noticed that the \SMfq polynomials over the small field involve a large number of linear variables compared to \SMfqm. This leads us to
propose a new \cref{modeling:final_fqm} to attack RD, which
relies on solving \SMfqm together with \MMfq. In this way, we take advantage of all the $\textstyle m\binom{n-k-1}{r}$ linear polynomials we can get in the $c_T$'s from \MMfq while keeping only $k$ linear variables $x_i$'s over $\ff{q^m}$ from \SMfqm . This increased compactness makes that even if this system were to be solved at
higher degree than \SMfq, it may perform better from a complexity
point of view. 

Let $ \NF(f,\langle P_{i,J}\rangle)$ be the normal form function that associates to any polynomial $f$
the unique polynomial $\widetilde{f} = f\mod \langle P_{i,J}\rangle $ such
that no $c_T$ leading term of a polynomial in the
$\langle P_{i,J}\rangle$ ideal appears in $\widetilde{f}$. \cref{modeling:final_fqm} is the system \eqref{eq:SMpfqm} over $\ff{q^m}$ which consists of the
  polynomials in $\mathcal Q_{\ge 2}$ in which the polynomials $P_{i,J}$'s
  are used to remove $c_T$ variables, i.e. $\lbrace \widetilde{Q_I},~Q_I \in \mathcal Q_{\ge 2} \rbrace$ where
  \begin{align*}
    \widetilde{Q_I}\eqdef \NF(Q_I, \langle P_{i,J}\rangle).
  \end{align*}
Then, we solve \cref{modeling:final_fqm} using the same technique as in \cite{BBCGPSTV20} by multiplying the polynomials by all possible monomials of degree $b-1$ in the $x_i$'s. Once again, the complexity analysis requires to estimate the dimension of the $\ff{q^m}$-vector space generated by the resulting bi-degree $(b,1)$ polynomials. According to~\cref{prop:Nb}, there are $\Nbfqm$ such polynomials but we provide in this section
new syzygies brought by the elimination of the $c_T$ variables using
the linear polynomials $P_{i,J}$. We call $\Nbfqsyz$ the
number of those new syzygies, so that the estimated dimension is $\Nbfqm - \Nbfqsyz$. The final cost follows by comparing this number to the number of monomials \Mbfq.

\begin{repproposition}{prop:recap}
  For any $b\ge 1$, the number of linearly independent polynomials at
  bi-degree $(b,1)$ in \SMpfqm is generically
  \begin{align*}
    \Nbfq &= \Nbfqm - \Nbfqsyz,
  \end{align*}
  with the exact value (from~\cref{prop:Nb})
  \begin{align}
    \Nbfqm &=               \sum_{i = 1}^{k} \binom{n-i}{ r}\binom{k+b-1-i}{b-1} 
             - \binom{n-k-1}{ r}\binom{k+b-1}{b} \tag{\ref{eq:Nbfqm}}
  \end{align}
  and the conjectured value, valid as long as $\Nbfq < \Mbfq$:
  \begin{align}
    \Nbfqsyz &=  (m-1)\sum_{i=1}^{b} (-1)^{i+1} 
    \binom{k+b-i-1}{b-i}\binom{n-k-1}{r+i}.
  \end{align}
  The number of monomials is
  \begin{align}
    \Mbfq &= \binom{k+b-1}{b}\left(\binom{n}{r} - m\binom{n-k-1}{r} \right),
  \end{align}
  so that we can solve \SMpfqm by linearization at bi-degree $(b,1)$  whenever
  \begin{align*}
    \Nbfq \ge \Mbfq-1.
  \end{align*}
  In this case, the final cost in $\fq$ operations is given by
  \begin{align*}
    \mathcal{O}\left( m^2\Nbfq {\Mbfq }^{\omega-1}\right),
  \end{align*}
  where $\omega$ is the linear algebra constant and where the $m^2$ factor comes from expressing each $\ff{q^m}$ operation involved in terms of $\ff{q}$ operations. 
\end{repproposition}
Note that it is always possible, whenever the ratio between polynomials and variables is much larger than 1, to drop excess polynomials by taking punctured codes much in the same way as in \cite[\S 4.2]{BBCGPSTV20}. 

\subsubsection{Analysis of the syzygies in~\cref{modeling:final_fqm}.} Contrary to Section \ref{sec:two}, we are not able to give a proof for the number of linearly independent syzygies due to the $P_{i,J}$'s. This comes from the fact that now, for some large enough $b$, we can solve the system, implying that the polynomials are not linearly independent at this degree anymore (hence we cannot give a general proof of independence). Also, we may find specific instances for which our conjecture fails. Still, we can analyse the generic behaviour on random
instances. Here, we describe generic syzygies coming from the $\ff{q^m}$
structure and we use them to count precisely the number of polynomials
and monomials at each bi-degree $(b,1)$ to determine the success of a
solving strategy by linearization in the generic case.

We start by giving a generalization of~\cref{prop:Q0}, that provides an explanation for the relations between the $\widetilde{Q}_I$ polynomials starting at bi-degree $(1,1)$.

\begin{repproposition}{prop:QIofPiJ}
  For any $T\subset\Iint{1}{n-k-1},~\#T=r+1$ and
  $1\le i \le m$, we obtain a relation between the $\widetilde{Q}_I$ polynomials given by
  \begin{align}
    \tr(\beta_i^\star)\widetilde{Q}_{T+k+1} + \sum_{\substack{I\subset\Iint{1}{n}\\\#I = r+1\\I\cap\Iint{ k+1}{n}\subsetneq T+k+1}}\tr(\beta_i^\star \minor{\Hm_{\yv}}{T,I})\widetilde{Q}_I = 0.
  \end{align}
  Note that the coefficients of any of these relations belong to $\fq$.
\end{repproposition}
\begin{proof}
  This comes from the fact that, for any $0\le \ell\le m-1$: 
  \begin{align*}
    \Gamma_{\ell,T} &\eqdef  \minor{
    \begin{pmatrix}
      \xv\Gm + \yv\\\Cm
    \end{pmatrix}
    \trsp{(\pow{\Hm_{\yv}}{\ell})}}{\any,T} = 0 \mod \langle P_{i,J}\rangle.
  \end{align*}
  Further details as well as the
  link between $P_{J}^{[\ell]}$ and $P_{i,J}$ are postponed in~\cref{proof:QIofPiJ}.  \qed
\end{proof}
\cref{prop:QIofPiJ} gives (at most) $\textstyle m\binom{n-k-1}{r+1}$ syzygies at bi-degree $(1,1)$ which include the relations from~\cref{prop:Q0} (the $\ell=0$ case in the proof).

At degree $b=2$, those relations multiplied by all linear
variables generate new relations, but they are not independent anymore:
indeed, for $1\le \ell \le m-1$ and any
$T_2\subset\Iint{1}{n-k-1},~\#T_2=r+2$ the following minor
gives $\textstyle (m-1)\binom{n-k-1}{r+2}$ relations between the
$ \Nbfqsyz[1]$ syzygies at bi-degree $(1,1)$:
\begin{align*}
  \minor{
	\begin{pmatrix}
		\xv\Gm + \yv\\ 	\xv\Gm + \yv \\ \Cm
	\end{pmatrix}
  \trsp{(\Hm_{\yv}^{[\ell]})}}{\any,T_2} = 0,
\end{align*}
More generally, a similar inclusion-exclusion combinatorial argument as those used to derive \cite[Heuristic 2]{BBCGPSTV20} leads to the following Conjecture \ref{conj:syz_PiJ}, that was verified experimentally for $b=2,~b=3$ and $b=4$. 
\begin{conjecture}\label{conj:syz_PiJ}
  For $b \geq 1$, the number of independent syzygies is expected to be equal to
  \begin{equation*}\label{eq:Nsyzb}
    \Nbfqsyz=(m-1)\sum_{i=1}^{b} (-1)^{i+1} 
    \binom{k+b-i-1}{b-i}\binom{n-k-1}{r+i}.
  \end{equation*}
\end{conjecture}

\section{Hybrid technique on minor variables}
\label{sec:hybridMM}
In algebraic cryptanalysis, ``hybrid approach'' usually refers to a generic method to possibly decrease the complexity of an algebraic attack by (a) choosing a subset of unknowns, (b) specializing them to some value, (c) solving the new system with less unknowns and (d) finally trying all possible specializations of those unknowns. The point is that in certain cases, the complexity gain in solving the new system supersedes the loss in complexity coming from exhaustive search. In \cite{BBCGPSTV20}, an indirect approach is followed on the MaxMinors modeling. Instead of performing a naive exhaustive search on random minor variables, the authors proceed by fixing $a \geq 0$
columns in $\Cm$. It can readily be seen that this provides $\textstyle{N\eqdef \binom{n}{r}-\binom{n-a}{r}}$ linear polynomials involving the $c_T$'s. These polynomials can in turn be used to reduce the number of $c_T$ variables by the same amount and this costs only to test $q^{a \cdot r}$ different choices instead of trying $q^N$ choices if we had performed the naive exhaustive search on $N$ variables. 

We show here that a variation of this idea, namely if we can fix $a$ columns of $\Cm$ to $0$, or basically what amounts to the same, if we can fix to $0$ $a$ positions of the error $\ev$ we seek in the RD problem, can have a dramatic effect on the Support-Minors modeling. Not only do we have the aforementioned reduction in the $c_T$ variables, but we do have a reduction of the number of linear variables as well. Moreover, the effect of this hybrid approach is even independent from the algebraic modeling or algorithm we use to solve the MinRank/RD problem in the sense that this hybrid approach actually provides a reduction to a smaller MinRank/RD problem. More precisely (we give here the explanation just for the RD problem):
\begin{enumerate}
\item \label{item:1} If by chance $a$ positions of the error vector are zero and the $a$ positions belong to an information set of the code, it is possible to reduce the problem with parameters $(m,n,k,r)$ to a smaller instance with parameters $(m,n-a,k-a,r)$; 
\item This has a chance $\textstyle{\frac{1}{q^{ar}}}$ to happen for a random instance;
\item It is possible to change the initial instance into an instance satisfying Point \ref{item:1}, either by using a deterministic search among all $q^{ar}$ possible transformations, or by using a rerandomizing trick that will succeed with probability $\textstyle{O(q^{-ar})}$.
\end{enumerate}
The idea to look for a particular error with zero positions is used in
\cite[\S 5.2]{GRS16}, where the rerandomizing trick is implicit (see
the proof of Proposition 3 there). Here, we present a way to reduce
the solving of an RD instance to a smaller problem when the error
vector is zero on some positions. The advantage is that the method is
applicable to any algorithm solving RD.

As explained above,  this is more general and it actually applies to any
MinRank problem. The rerandomizing trick applies equally to both cases and we begin our discussion by explaining it. The proofs are somewhat simpler in the RD case and we start with this more specific case before turning to the MinRank case. We end the section with a probabilistic description of the rerandomization trick.

\subsection{Rerandomizing the MinRank and the RD instances}\label{ss:rerandomizing}
There is no reason a priori why $a$ positions of the RD solution $\ev$   or $a$ columns of the MinRank solution $\textstyle{\Em = \Mm_0 + \sum_{i=1}^K \Mm_i}$ would  be equal to $0$. The point is that 
we can multiply on the right $\ev$ or $\Em$ by an invertible $n \times n$ matrix $\Pm$ with coefficients over $\Fq$. This does not change the rank weight of $\ev$ or the rank of $\Em$, but now
$a$ positions or $\ev$ or $a$ columns of $\Em$ have a chance to be equal to $0$. Moreover, if we make the following assumption on the $(m,n,k,r)$-RD instance $(\yv,\Cc,r)$ or the 
$(m,n,K,r)$ MinRank instance $(\Mm_0,\Mm_1,\dots,\Mm_K,r)$,
\begin{assumption}\label{ass:basic}
In the RD case, we assume that the first $r$ positions of the solution $\ev$ are independent over $\Fq$. In the MinRank case, we assume that 
the first $r$ columns of the solution $\textstyle{\Em=\Mm_0 + \sum_{i=1}^K x_i \Mm_i}$ are independent.
\end{assumption}
 then 
we can even try at most $q^{ar}$ matrices belonging to the set 
{\small

\begin{equation}\label{eq:set_P}
  \cP \eqdef \left\{ \Pm_{\Am} = 
  \begin{pmatrix}   
  \ident_{r} & \zerom_{r \times (n-a-r)} & - \Am \\
\zerom_{(n-a-r) \times r} & \ident_{n-a-r} & \zerom_{(n-a-r) \times a} \\
\zerom_{a \times r} & \zerom_{a \times (n-a-r)} & \ident_{a} 
\end{pmatrix}
,~\Am \in \Fq^{r \times a} \right\}.
\end{equation}}

The point is that
 multiplying by matrices of this form amounts to leave the $(n-a)$  columns in the first two blocks unchanged, but adds to the  last $a$ positions of $\ev$  or the last $a$ columns  of $\Em$ all possible linear combinations of the $r$ first ones. One of them has to be $0$ because by assumption, the $r$ first positions/columns form a basis of the subspace $\vsg{e_1,\dots,e_n}{\Fq}$ or the column space of $\Em$.
 One could think that this would give a new instance of the MinRank problem associated to the matrices $\Mm'_0 = \Mm_0 \Pm_{\Am}$, $\Mm'_1 = \Mm_1 \Pm_{\Am},\dots , \Mm'_K = \Mm_K \Pm_{\Am}$, 
 or an RD problem associated to the word $\yv' = \yv \Pm_{\Am}$ and the code 
 $\Cc_{\Am} = \{ \cv \Pm_{\Am}: \cv \in \Cc\}$, however the fact that the last columns are equal to $0$ has an additional effect, we can namely reduce accordingly the dimension of the matrix code (in the rank metric case) or of the underlying $\Fqm$-linear code. Let us verify this in the RD case first. We are going now to use the following notation in the subsections that follow
 \begin{eqnarray}
 J & \eqdef & \Iint{n-a+1}{n} \label{eq:J}\\
 \Jch & \eqdef & \Iint{1}{n-a}. \label{eq:Jch}
 \end{eqnarray}

\subsection{RD instances}\label{sec:hyb_rd}
It turns out that in RD case, when $\left( \ev \Pm_{\Am} \right)_J = \zerom_{a}$, with a very mild condition on the shortened code at $J$ we obtain a reduction to an RD instance with smaller parameters, namely
\begin{proposition}\label{prop:shJ}
Assume that $\ev_J = \zerom$. Let $\sh{J}{\Cc}$ be the code $\Cc$ shortened at $J$, namely
$\sh{J}{\Cc}  = \{\cv_{\Jch}:\; \cv \in \Cc, \;\cv_J = \zerom\}$. Assume that $\sh{J}{\Cc}$ is of dimension $k-a$,  then by Gaussian elimination on a generator matrix $\Gm$ of $\Cc$ we can obtain a generator matrix of $\Cc$ in systematic form on the columns in $J$, i.e.
\begin{align*}
 \Dm  \Gm &= \bordermatrix{ & \Jch & J \cr
  & \Gm' & \zerom_{(k-a)\times a}\cr
           &\Bm & \Im_a}
\end{align*}
for some invertible matrix $\Dm\in \Fqm^{k\times k}$. Then $\Gm'$ is a generator matrix of $\Cc'\eqdef\sh{J}{\Cc}$. Define $\yv'\eqdef (\yv)_{\Jch} - \yv_J\Bm$. Then $(\yv',\Cc',r)$ is a valid instance of an RD problem of parameters $(m,n-a,k-a,r)$, from which we can deduce a solution of the initial problem $(\yv,\Cc,r)$.
\end{proposition}

\begin{proof}
The first point is just standard linear algebra. For the second point, let $(\cv,\ev=\yv-\cv)$ be  the solution  of the RD problem and 
denote by $(\xv',\xv'')$ where $\xv' \in \Fq^{k-a}$ and $\xv'' \in \Fq^a$ the vector defined by
\begin{eqnarray*}
(\xv',\;\xv'') & = & \xv \Dm^{-1}\;\text{ where }\\
\cv & = & \xv \Gm. 
\end{eqnarray*} 
Observe now that
\begin{eqnarray*}
\ev_J & = & \yv_J - \cv_J \\
& = & \yv_J -\left( \xv \Gm \right)_J\\
& = & \yv_J - \left( \left( \xv',\;\xv''\right)\Dm \Gm \right)_J \\
& = & \yv_J - \xv''.
\end{eqnarray*}
Since $\ev_J=\zerom_{a}$, this implies $\yv_J=\xv''$. Therefore
  \begin{align*}
    \ev_{\Jch} = \yv_{\Jch}-\cv_{\Jch} &= \yv_{\Jch} - \xv'\Gm' - \xv''\Bm\\
    & = \underbrace{\yv_{\Jch}- \yv_J\Bm}_{\yv'} - \underbrace{\xv'\Gm' }_{=\cv'\in \Cc'}
  \end{align*}
Therefore $\yv' - \cv'$ is of rank weight $r$ and the proposition follows.
\qed
\end{proof}

This proposition is used as follows. When we want to solve an instance $(\yv,\Cc,r)$ of an RD problem of parameters $(m,n,k,r)$, we consider $q^{ar}$ RD instances $(\yv\pri,\Cc\pri,r)$ of parameters $(m,n-a,k-a,r)$ obtained from all $\Pm_{\Am} \in \cP$ by computing a generator matrix $\Gm_{\Am} \eqdef \Gm \Am$ (where $\Gm$ is a generator matrix of $\Cc$) of the code
$\Cc_{\Am} \eqdef \{\cv \Pm_{\Am}:\;\cv \in \Cc\}$, then put this matrix in (partial) systematic form on the columns in $J$ by Gaussian elimination and obtain 
\begin{equation}
\label{eq:condition}
\Gm\dpr =  \bordermatrix{ & \Jch & J \cr
  & \Gm\pri& \zerom_{(k-a)\times a}\cr
           &\Bm & \Im_a}.
           \end{equation}
           The RD instances $(\yv\pri,\Cc\pri,r)$ are defined from $\Gm\pri$, $\yv$ and $\Pm_{\Am}$ by
\begin{eqnarray*}
\Cc\pri & \eqdef & \{\xv \Gm\pri:\;\xv \in \Fqm^{k-a}\}\\
\yv\pri &\eqdef & \yv\dpr_{\Jch} -\yv\dpr_J \Bm,\text{ where}\\
\yv\dpr &\eqdef & \yv \Pm_{\Am}.
\end{eqnarray*}
Finally, one of these instances has a solution from which we recover the solution of the original problem thanks to Proposition \ref{prop:shJ}.

It remains to check under which condition we can put $\Gm_{\Am}$ in partial systematic form for any $\Am \in \Fq^{r \times a}$ as required in \eqref{eq:condition}. This is given by Proposition \ref{prop:shJ}: namely that $\sh{J}{\Cc_{\Am}}$ should have  dimension $k-a$ for any $\Am$. There are two cases to consider:

\noindent
{\bf Case 1:} $a+r \le k$.\\
In this case, there is a very mild condition on $\Cc$ for which the relevant property holds for any $\Am \in \Fq^{r \times a}$, namely that
\begin{lemma}\label{prop:r+a<=k}
  Provided that there exists a systematic set for $\Cc$ that contains $\{1,\dots,r\} \cup J$,  the code $\sh{J}{\Cc_{\Am}}$ has dimension exactly $k-a$ for all $\Am\in\Fq^{r\times a}$. 
\end{lemma}
\begin{proof}
By reordering the positions we may assume that the systematic set is $\Iint{1}{k}$ and $J=\Iint{r+1}{r+a}$ and 
$$\Pm_{\Am} = \begin{pmatrix}   
  \ident_{r} &  - \Am& \zerom_{r \times (n-a-r)}  \\
  \zerom_{a \times r} & \ident_{a}  & \zerom_{a \times (n-a-r)} \\
\zerom_{(n-a-r) \times r} &  \zerom_{(n-a-r) \times a} &  \ident_{n-a-r} 
\end{pmatrix}.$$
On the other hand we can assume by the hypothesis of the lemma that we can choose the generator matrix of $\Cc$ as 
$$
\Gm = \begin{pmatrix} 
\Im_k & \Rm
\end{pmatrix}.
$$
The generator matrix of $\Cc_{\Am}$ is of the form
$$
\Gm \Pm_{\Am} = \begin{pmatrix}   
  \ident_{r} &  - \Am& \zerom_{r \times (n-a-r)}  &  \Rm_1\\
  \zerom_{a \times r} & \ident_{a}  & \zerom_{a \times (k-a-r)} & \Rm_2 \\
\zerom_{(k-a-r) \times r} &  \zerom_{(k-a-r) \times a} &  \ident_{k-a-r}  & \Rm_3
\end{pmatrix}.
$$
This code $\Cc_{\Am}$ is therefore still systematic in the first $k$ positions and hence $\sh{J}{\Cc_{\Am}}$  has dimension exactly $k-a$.
\qed
\end{proof}

\noindent
{\bf Case 2:} $r+a >k$.\\
Note that in this case, the $\ff{q^m}$-linear code $\mathcal{D}$ of parameters $[r+a,k]$ which is generated by the matrix $\Gm_{*,\Iint{1}{r} \cup J} \in \ff{q^m}^{k \times (r+a)}$ is not the full code. 
It is also worthwhile to notice that  $\sh{J}{\Cc_{\Am}}$ has dimension $k-a$ if and only
if the matrix
$\Gm_{\any,J}-\Gm_{\any,\Iint{1}{r}}\Am$ has rank
$a$. To verify whether or not this property holds for any $\Am$ we use the following lemma.

\begin{lemma}\label{prop:r+a>k}
	The existence of a matrix $\Am \in \ff{q}^{r \times a}$ such that $\Gm_{*,J}-\Gm_{*,\Iint{1}{r}}\Am$ is rank defective is equivalent to the existence of a word of weight $ \leq a$ whose support is spanned by the $a$ last coordinates in the dual of $\mathcal{D}$.
\end{lemma}
\begin{proof}
	Assume that some $\Am \in \ff{q}^{r \times a}$
	is such that $\rank\left(\Gm_{*,J}-\Gm_{*,\Iint{1}{r}}\Am \right)<a$. This means that there exists a vector $\boldsymbol{\lambda}_{\Am} \in \ff{q^m}^a$ such that
	\begin{equation*}
	 - \Gm_{*,\Iint{1}{r}}\Am\trsp{\boldsymbol{\lambda}_{\Am}}+\Gm_{*,J}\trsp{\boldsymbol{\lambda}_{\Am}} =  \Gm_{*,\Iint{1}{r} \cup J}\underbrace{\begin{pmatrix}
 -\Am \trsp{\boldsymbol{\lambda}_{\Am}}\\		\trsp{\boldsymbol{\lambda}_{\Am}} \end{pmatrix}}_{:=\trsp{\vv_{\Am}}} = 0.
	\end{equation*}
	In particular, the vector $\vv_{\Am} \in \ff{q^m}^{a+r}$ belongs to $\dual{\mathcal{D}}$, its weight is $\leq a$ (as the entries of $\Am$ belong to $\fq$) and its support is spanned by the $a$ last coordinates. The converse statement is similar by constructing an inverse of the map $\Am \mapsto \vv_{\Am}$. \qed
      \end{proof}
Under the assumption that $\mathcal{D}$ behaves as a random code with parameters $[a+r,k]$, one can show that
\begin{proposition}
The probability that there exists in the dual of a random $\Fqm$-linear code of parameters $[a+r,k]$  a non zero codeword of weight $\leq a$ whose support is spanned by the $a$ last coordinates is upper-bounded by $\Theta\left(q^{(m+r)a-mk}\right)$ as $q\to \infty$.
\end{proposition}
\begin{proof}
This probability is upper-bounded by the probability that there exists simply a non zero codeword of weight $\leq a$ in such a code. Let $X$ be the number of such codewords. We use the fact that $\Prob{X \neq 0} \leq \esp(X)$ and that the
expected number $\esp(X)$ of non-zero vectors of weight $\leq a$ in such a code is given by
\begin{eqnarray*}
\esp(X) & = & \frac{B_a - 1}{q^{mk}},
\end{eqnarray*}
where $B_a$ is the size of a ball of radius $a$ in $\Fqm^{a+r}$ in the rank metric. By using 
\cite[Proposition 1]{L14a} the size of such a ball is of the form
$\Theta\left(q^{(m+a+r)a-a^2}\right)=\Theta\left(q^{(m+r)a}\right)$ for any nonnegative integer $a \leq m$.
We deduce the proposition from this. \qed
\end{proof}

\subsection{MinRank instances}\label{sec:hyb_minrank}
This reduction sketched for the RD problem also applies to MinRank. 
Consider a MinRank instance $(\Mm_0,\dots,\Mm_K)$ with target rank
$r$, and denote by $\textstyle{\Em=\Mm_0+\sum_{i=1}^K x_i\Mm_i}$ the rank $r$ matrix we are looking for. To explain the form taken by the reduced RD instances we got in Subsection \ref{sec:hyb_rd}, it was convenient to put the
generator matrix of the transformed code $\Cc_{\Am}=\Cc \Pm_{\Am}$ into systematic form. It will be helpful here to use a similar notion in the MinRank case by viewing a matrix as the vector formed by the concatenation of its rows.  To define the relevant systematic form we will use, 
we bring in the invertible linear map
\begin{align}
  \varphi:  \Fq^{m\times n}& \to \Fq^{mn}\\
  \Am & \mapsto   (\Am_{i,j})_{i\in\Iint{1}{m},j\in\Iint{1}{n}}\notag{}
\end{align}
where the image of $\varphi(\Am)$ is formed by the entries of $\Am$ in
column-major order (we could equivalently take the row-major order).
Using $\varphi$ we define the generator matrix associated to a MinRank instance as follows.
\begin{definition}
  Let $\Mm_1,\dots,\Mm_K$ be $K$ matrices in $\Fq^{m\times n}$, and define $\mathcal L$ the matrix code generated by the $\varphi(\Mm_i)$'s. Then the following matrix $\Lm$ is a $K\times mn$ generator matrix of $\mathcal L$:
  \begin{align*}
    \Lm(\Mm_1,\dots,\Mm_K) &\eqdef
          \begin{pmatrix}
            \varphi(\Mm_1)\\\vdots\\\varphi(\Mm_K)
          \end{pmatrix}\in\Fq^{K\times mn}.
  \end{align*}
\end{definition}
As noted in \cite[\S 4.4]{BESV22}, any elementary row operation on $\Lm$ corresponds to linear transformations of the variables $x_i$, i.e. we can always transform the initial MinRank instance to an equivalent one with $\Lm$ in echelon form. From now on, we will assume that $\Lm$ is in echelon form.
\begin{definition}
  We say that MinRank instance is in systematic form if its associated generator matrix is. We denote by $S$ the systematic positions. 
  \end{definition}
  It is clear that
  \begin{fact}
  \label{fa:easy}
  If the MinRank instance is in systematic form, we can equivalently reduce $\varphi(\Mm_0)$   w.r.t. the generator matrix, then it has $K$ zeros in positions belonging to $S$. In this case, $\textstyle{\varphi(\sum_{i=0}^K x_i\Mm_i)}$
   contains $K$ consecutive positions equal to $(x_i)_{i \in S}$, i.e. the $K$ entries of the matrix $\Em$ belonging to $S$ are exactly the $K$ corresponding linear variables.
\end{fact}
\begin{remark}
  It is not always possible to put a MinRank instance in systematic form, as not any permutation of columns in $\Fq^{nm}$ preserves the rank (the permutation needs to permute blocks of columns in the corresponding matrix).
  But as noted in \cite{BESV22}, a random MinRank instance will be in systematic form with high probability.
  
\end{remark}

We use the same notation as in \eqref{eq:J} and \eqref{eq:Jch} for $J$ and $\Jch$ and denote by $I$ the set of positions of $\Iint{1}{mn}$ that correspond to the 
columns indexed by the positions in $J$, that is $I = \cup_{j\in J}\Iint{(j-1)m+1}{jm}$. Following the approach in \cite[Prop. 3]{GRS16}, we first
analyze the complexity of solving the MinRank instance with the 
columns in $J$ specialized to zero. We will then see how we can reduce to
this case, by using either a deterministic, or a probabilistic
approach.
\begin{proposition}[Assuming the error is zero on coordinates in $J$]
  \label{prop:minrankZero}
  Consider a MinRank instance $(\Mm_0,\dots,\Mm_K)$ in $\Fq^{m\times n}$
  with target rank $r$.  Assume that $am\le K$ and that the solution
  $\xv$ satisfies $\Em_{\any,J}=\zerom_{m\times a}$, or equivalently
  $\varphi(\Mm_0)_{I}+\xv\Lm_{\any,I}=\zerom_{am}$. Let
  $\mathcal L'\eqdef \sh{I}{\mathcal L}$ be the code $\mathcal L$
  shortened at $I$.  Assume that $\sh{I}{\mathcal L}$ is of
  dimension $K-am$, then a solution $\xv$ for $(\Mm_0,\dots,\Mm_K)$
  with target rank $r$ can be deduced from the solution of a smaller
  MinRank instance $(\Mm_0',\dots,\Mm_{K-am}')$ in
  $\Fq^{m\times (n-a)}$ with target rank $r$.

  More precisely, by Gaussian elimination on $\Lm$ we can obtain a
  generator matrix of $\mathcal L$ in systematic form on the columns
  in $I$, i.e. after permuting positions, so that the last positions belong to $I$:
  \begin{align*}
     \Dm  \Lm &= \begin{pmatrix}
   \Lm' & \zerom_{(K-am)\times {am}} \\
           \Bm & \Im_{am}
           \end{pmatrix}
  \end{align*}
  for some invertible matrix $\Dm\in\Fq^{K\times K}$. Then $\Lm'\in\Fq^{(K-am)\times m(n-a)}$ is a generator matrix of   $\mathcal L'$. Define $\Mm_i'$ to be the $m\times (n-a)$ matrix corresponding to the $i$-th row in $\Lm'$, and\footnote{We abusively use the same name $\varphi:\Fq^{m\times n}\to\Fq^{mn}$ and $\Fq^{m\times (n-a)}\to\Fq^{m(n-a)}$.} $\Mm'_0=\varphi^{-1}(\varphi(\Mm_0)_{\Ich}-\varphi(\Mm_0)_{I}\Bm)$ of size $m\times (n-a)$, where $\Ich \eqdef \Iint{1}{mn} \setminus I$. Then $\Mm'_0,\Mm'_1,\dots,\Mm'_{K-am}\in\Fq^{m\times (n-a)}$ is a  MinRank instance with target rank $r$, and any solution $\xv'$ of this instance gives a solution $\xv=\Dm(\xv'\;\xv'')$ of the initial instance with $\xv''=-\varphi(\Mm_0)_{I}$.
\end{proposition}
\begin{proof}
To simplify the explanations, we assume that the positions in $\Iint{1}{mn} $ have been permuted, so that the last $am$ positions belong to $I$.
By hypothesis, we have $\textstyle{\Dm\Lm_{\any,I}=\binom{\zerom}{\Im_{am}}}$, so that if $\xv\Dm^{-1}=(\xv'\;\xv'')$ with $\xv'$ of size $K-am$, the hypothesis $\varphi(\Mm_0)_{I}+\xv\Lm_{\any,I}=\zerom$ is equivalent to $\xv''+\varphi(\Mm_0)_{I}=\zerom$. As $\Em_{J}=\zerom$, the matrix $\Em_{\Jch}$ has rank $r$, and is given by
  \begin{align*}
\varphi(\Em_{\Jch})=    \varphi(\Em)_{\Ich} &= \xv'\Lm'+\xv''\Bm + \varphi(\Mm_0)_{\Ich}\\
                        &= \xv'\Lm' - \varphi(\Mm_0)_{I}\Bm + \varphi(\Mm_0)_{\Ich}.\\
\text{i.e.} \;    \Em_{\Jch}&= \Mm_0' + \sum_{i=1}^{K-am} x_i'\Mm_i'.
  \end{align*}
  We get the smaller MinRank instance described in the proposition.
  \qed
\end{proof}
\paragraph{A deterministic way to reduce to the zero case.}
Similarly to the RD case,  we can reduce a MinRank problem of parameters $(m,n,K,r)$ to solving $q^{ar}$ MinRank instances of parameters $(m,n-a,K-am,r)$ obtained by multiplying the $\Mm_i$'s by $\Pm_{\Am}$ in $\cP$, then under the assumption that all shortened codes are of rank $K-am$ we can apply Proposition \ref{prop:minrankZero} to them.
Eventually, exactly one of the resulting 
$(\Mm_0 \Pm_{\Am},\dots,\Mm_K\Pm_{\Am})$ instances will have a solution that is zero on the columns $J$, and then leads to the desired solution. 
We give here an example where it is always true provided that $(r+a)m\le K$.
  \begin{lemma}\label{lemma:minranksystform}
    Assume the MinRank instance is in systematic form on a set of positions $S$ that contains $\Iint{1}{rm}\cup I$,  
  then $\sh{I}{\Lm_{\Am}}$ has rank $K-am$ for all $\Am\in\Fq^{r\times a}$.
  \end{lemma}
  \begin{proof}
    The matrices $\Mm_i^{\Am}\eqdef \Mm_i\Pm_{\Am}$ are identical to
    $\Mm_i$ on the columns $\Jch$, and the columns in $J$ are
    $(\Mm_i^{\Am})_{\any,J}=(\Mm_i)_{\any,J}-(\Mm_i)_{\any,\Iint{1}{r}}\Am$.
    We reorder the positions so that the systematic positions are the
    $K$ first ones and such that $I = \Iint{rm+1}{(r+a)m}$.  If the
    MinRank instance is in systematic form then for $i\in\Iint{1}{K}$
    such that $i=(v-1)m+u$ with $v\in\Iint{1}{n}$ and
    $u\in\Iint{1}{m}$, we have that $ (\Mm_i)_{\any,\Iint{1}{r+a}}$
    has at most only one nonzero entry 1 in position $(u,v)$ if
    $v\le r+a$, and is all zero otherwise. This means that
    \begin{align*}
(\Mm_i)_{\any,\Iint{1}{r}}=\zerom_{m\times r} \text{ hence }      \Mm_i^{\Am}=\Mm_i & \text{ for } i\ge rm+1,
    \end{align*}
    and that
    \begin{align*}
      (\Mm_i^{\Am})_{\any,J}
      &=
 \begin{pmatrix}
            \zerom\\
            -\Am_{v,\any}\\\zerom
            \end{pmatrix} \leftarrow \text{ row } u
&\text{ for } i\in\Iint{1}{rm}.
    \end{align*}
Finally, this means that
\begin{align*}
  \Lm_{\Am} &=
              \bordermatrix{&&\text{\small positions in } I\cr
                               &\Im_{rm} & \left(\substack{\text{coefficients}\\\text{ depending}\\\text{ on } \Am}\right) &\zerom &  \Lm_{\Iint{1}{rm},\Iint{(a+r)m+1}{nm}}\cr
                                                                                                                 &\zerom & \Im_{am} & \zerom & \Lm_{\Iint{rm+1}{(r+a)m},\Iint{(a+r)m+1}{nm}}\cr
                &\zerom & \zerom & \Im_{K-(a+r)m} & \Lm_{\Iint{(r+a)m+1}{K},\Iint{(a+r)m+1}{nm}}
              }
\end{align*}
is full rank on the columns in $I=\Iint{rm+1}{(r+a)m}$, i.e. that $\sh{I}{\Lm_{\Am}}$ has rank $K-am$.
    \qed
  \end{proof}

\subsection{Probabilistic hybrid approach on  MinRank or RD instances}
\label{sec:hybridproba}

The approach given in the previous sections is a deterministic way to solve generic MinRank or RD instances. However, it does not work if the initial conditions on the solution $\Em$ of the MinRank problem or on the solution $\ev$ of the RD problem are not met, i.e. the first $r$ columns of $\Em$ are not linearly independent, or the first $r$ entries of $\ev$ are not linearly independent over $\Fq$. This can be fixed by considering instead a randomized algorithm, which consists in multiplying on the right the MinRank instance by a random $n \times n$ invertible
matrix $\Pm$ over $\Fq$ which gives with probability $\Om{1}$ a new instance of the $(m,n,K,r)$ MinRank problem which satisfies all the right assumptions and on which we can apply the aforementioned technique. Once we have solved the new MinRank problem, we recover the solution of the original MinRank by multiplying it on the right by $\Pm^{-1}$. A similar technique can also be used for the RD problem. This might even be improved slightly by multiplying on the right each time by a new $\Pm$ and making directly the bet that $\Em \Pm$ has all its $a$  columns in $J$ equal to $0$ (i.e we assume directly that we have an instance of the $(m,n-a,K-am,r)$ problem). This has a probability of 
$\Om{q^{-ar}}$ to happen. In both cases, we get a probabilistic algorithm of similar complexity as the deterministic algorithm, with the difference that it would work on {\em any} $(m,n,K,r)$ instance of
the MinRank problem.
\subsection{Complexity of the hybrid technique}
\label{sec:hyb_complexity}
Let $\mathcal A$ be an algorithm that solves the MinRank problem, and $T_{\mathcal A,\text{plain},(m,n,K,r)}$ its cost on a generic MinRank problem of parameters $(m,n,K,r)$.
In the MaxMinors case, the original purpose of fixing columns in $\Cm$ was to end up with an overdefined linear system. Here, fixing $a \geq 0$ columns  yields to the solving of a smaller problem.
Therefore, the cost of the hybrid technique is estimated by solving a minimization problem over $a \geq 0$.
Under the assumption 
that the resulting MinRank instances of parameters $(m,n-a,K - a m,r)$ behave as random, we have

\begin{proposition}\label{prop:cost_hybrid_sm}
	The time complexity of the proposed  hybrid technique on a generic MinRank problem of parameters $(m,n,K,r)$ is given by
	\begin{equation*}
T_{\mathcal A, \text{hybrid},(m,n,K,r)}  = \min_{a \geq 0}\left( q^{a  r}\cdot T_{\mathcal A, \text{plain},(m,n-a,K - a m,r)} \right).
	\end{equation*}
      \end{proposition}
We may obtain a similar statement in the RD case, where we consider any  algorithm $\mathcal{A}$ to solve RD.
\begin{proposition}\label{prop:cost_hybrid_rd}
	The time complexity of the proposed hybrid technique applied to an algebraic algorithm $\mathcal{A}$ to solve an RD problem of parameters $(m,n,k,r)$ is given by
	\begin{equation*}
		T_{\mathcal{A},\text{hybrid},(m,n,k,r)}  = \min_{a \geq 0}\left( q^{a  r}\cdot T_{\mathcal{A},\text{plain},(m,n-a,k - a,r)} \right).
	\end{equation*}
\end{proposition}

In particular, this applies to the new \SMpfqm approach presented in this paper. The overall complexity may be easily computed by combining Proposition \ref{prop:recap} to obtain $T_{\text{SM-}\ff{q^m}^{+},\text{plain}}$ with Proposition \ref{prop:cost_hybrid_sm}.

\section{Estimated costs on MinRank and RD instances.}
\label{sec:numerical}

Finally, we provide the bit complexity of the attacks described in this paper on some parameter sets. First, we apply the hybrid technique described in Section \ref{sec:hybridMM} to the Support-Minors modeling on generic MinRank instances (see Proposition \ref{prop:cost_hybrid_sm}). The same technique is then used on the \SMpfqm system from Section \ref{sec:solvingfqm} to attack RD instances (see Proposition \ref{prop:recap} and Proposition \ref{prop:cost_hybrid_rd}). In both cases, these attacks are compared to former attacks on MinRank and RD.
{The {\tt magma} code used to produce the Tables and Figures is available on \url{https://github.com/mbardet/Rank-Decoding-tools}.}

\subsection{MinRank instances}
\label{ss:numerical}

On the plain MinRank problem, the approach of Section \ref{sec:hybridMM} on Support-Minors allows to reach smaller complexities than the ones obtained with the specialization technique of \cite{BBCGPSTV20} which consists in fixing linear variables. More interestingly, our proposed hybrid approach actually offers a trade-off between combinatorial attacks (e.g. Goubin's Kernel attack \cite{GC00}) and pure algebraic attacks. Indeed, the bet that we make can be seen as guessing $a \geq 0$ vectors in the right kernel of the low rank matrix $\Mm$ similary to \cite{GC00}, the difference being that we consider less vectors than $\textstyle\lceil \frac{K}{m} \rceil$. 

As an illustration, we give the complexity of our attack on the parameters of the MinRank based signature scheme \cite{BESV22b} in Table \ref{tab:minrank_challenge_sm_hybrid} which builds upon the seminal work of Courtois \cite{C01}. 
{Note that the parameters proposed in~\cite{BESV22b} already take into account our improved MinRank attack.}

\begin{table}
	\centering
	\caption{\label{tab:minrank_challenge_sm_hybrid}
          Bit complexity of the Kernel attack and the hybrid SM attack on the parameters from \cite{BESV22b}. The number of guessed vectors in the Kernel attack is equal to $\textstyle a := \lceil \frac{K}{n} \rceil$ and the final complexity in $\ff{q}$-operations is $\mathcal{O}\left(q^{a\cdot r} K^{\omega} \right)$.  For the SM attack, we report the triplet $(b,a,n_{\text{cols}})$ which leads to the best complexity: ``$a$" refers to the number of guessed columns, $``n_{\text{cols}}"$ is the number of columns in the reduced MinRank problem ($\leq n-a$) and $b$ is the degree at which we solve via SM. Finally, we adopt $\omega = 2$ as in \cite{BESV22b}, a constant factor of $7$ in Strassen's algorithm and we consider that a multiplication over $\ff{2^4}$ represents 23 binary operations.
          We also report in this table the optimized kernel attack as given in \cite{BESV22b} which improves  on the polynomial factor in front of the complexity.} 
	\begin{tabular}{|*{2}{c|}|c|c|c|c|c|c|}
          \hline
          $(q,m,n,K,r)$ & $\lambda$
          &  Kernel $(a)$  &Kernel in~\cite{BESV22} $(a)$
          &  SM Section \ref{sec:hybridMM}  $(b, a, n_{\text{cols}})$
          \\ \hline
          $(16,16,16,142,4)$ & 128 & 166 (9)  &158 (8) & 161 (5, 6, $n-a$)\\ \hline
          $(16,19,19,167,6)$ & 192 & 238 (9)  &231 (8)& 231  (7,  6,  $n-a$) \\ \hline
          $(16,22,22,254,6)$ & 256 & 311 (12)  &303 (11) & 297  (1, 11, $n-a$) \\ \hline
	\end{tabular}
\end{table}

\subsection{RD instances}

Recall that the cost of the best combinatorial attack of \cite{AGHT18} in $\ff{q}$ operations is
 \begin{equation}\label{eq:combi}
 \mathcal{O}\left( (n-k)^{\omega} m^{\omega} q^{r\left\lceil\frac{(k+1)m}{n}\right\rceil-m} \right),
 \end{equation}
 where $\omega$ is the linear algebra constant. It is now common to take $\omega=2$: this value is optimistic, but take into account any algorithm that could take advantage of the structure of the matrices. 
 Also, cryptographically relevant RD instances are such that $r = \mathcal{O}(\sqrt{n})$ or such that the weight $r$ is closer to the Gilbert-Varshamov bound, and we selected parameter sets corresponding to these two situations. The $r = \mathcal{O}(\sqrt{n})$ regime is for instance the one encountered in the NIST submissions ROLLO and RQC. In Table \ref{tab:ROLLO1}, we give the binary logarithm of the complexity of our attack ``over $\ff{q^m}$" on ROLLO-I parameters and we also keep track of the optimum values of $a$ and $b$. This cost is compared to the one of the combinatorial attack of \cite{AGHT18} (``comb") and to the one of the MaxMinors attack (``MM.").

\begin{table}
  \centering
  \caption{\label{tab:ROLLO1} Comparison between known attacks on the new ROLLO-I parameters in~\cite{BBCGPSTV20} and \cite{AABBBBCDGHZ20} after the 2021-04-21 update. The ``*''-symbol means that the best attack is obtained on the derived code from  key attack with parameters $(m,n,k,r) = \textstyle{ (m,2k-\left\lfloor\frac{k}d\right\rfloor,k-\left\lfloor\frac{k}d\right\rfloor,d)}$, where $d$ refers to the rank of the moderate weight codewords in the masked LRPC code. Otherwise, the attack is on an RD problem with parameters $(m,2k,k,r)$. The struck out numbers are the underestimated values from \cite[Table 3]{BBCGPSTV20}. We also adopt $\omega=2$, whereas previous values where computed with $\omega=2.81$.}
  \begin{tabular}{|l|c|c|c|c|c||c|c|c||c|c|c||c|c|}
    \hline
    Instance & $q$ & $k$ & $m$ & $r$ & $d$ &  \MMfq & $a$ & $p$ & \SMpfqm & $b$&$a$ &comb\\ 
    \hline
    {\tiny new2ROLLO-I-128} & 2 & 83 & 73 & 7 & 8  & 205 & 18& 0& \sout{180} {\bf 202}  &2 & 13 &  212  \\
    \hline
    {\tiny new2ROLLO-I-192} & 2 & 97 & 89 & 8 & 8  & 226*& 17 & 0&\sout{197*} {\bf 223*}   &1&14& { 282*}  \\
    \hline
    {\tiny new2ROLLO-I-256} & 2 & 113 & 103 & 9 & 9 & 371* & 30 &1&\sout{283*}  {\bf 366*}  &1&27& { 375*}  \\
    \hline
    \hline
    {\tiny ROLLO-I-128-spe} & 2 & 83 & 67 & 7 & 8  & 212 & 19& 0& 214 &2& 15 & {\bf 196}  \\
    \hline
    {\tiny ROLLO-I-192-spe} & 2 & 97 & 79 & 8 & 8  & 242*& 19 & 0&{\bf 241*}  &2&15&{ 251*} \\
    \hline
    {\tiny ROLLO-I-256-spe} & 2 & 113 & 97 & 9 & 9 & 380* & 31 &0&{376*} &2&27& {\bf 353*} \\
    \hline
  \end{tabular}
\end{table}

\cref{fig:graph_MM_combi,fig:graph_MM_combiq} contain a broader comparison between the same attacks for fixed $(m,n,k)=(31,33,15)$ and  weight $r$ between $2$ and $d_{\text{RGV}}=10$, for $q=2$ in \cref{fig:graph_MM_combi} and $q=256$ in \cref{fig:graph_MM_combiq}. We can see that for $q=2$, the algebraic attacks become less efficient than combinatorial attacks for large $r$.
This justifies the current trend for rank-based proposals to  now consider a different regime where the weight $r$ is chosen closer to the rank Gilbert-Varshamow bound $d_{\text{RGV}} = \mathcal{O}(n)$, see for instance \cite{AADGZ22,BBBG22}.
Note also that in the scheme of \cite{AADGZ22} which uses LRPC codes, choosing $d$ of the same order as $r$ somehow increases the rank of the moderate weight codewords in the masked LRPC code and therefore may allow to gain confidence in the indistinguishability assumption.
This behavior can be explained by the fact that for the combinatorial attacks, the exponential part of the complexity all depends on $q$, whereas for the \MMfq{} attack, the cost $\textstyle{q^{ar}\binom{n-a}{r}^\omega}$ contains a part depending on $q$ whereas the other part $\textstyle{\binom{n-a}{r}^\omega}$ does not depends on $q$. This is the same for \SMpfqm{}. 

\pgfplotstableread{
            x      y        label
           2       19    MM 
            3       26    MM 
           4       33    MM 
           5       48    MM 
           6       70    MM 
           7       96    MM 
           8       114    MM 
           9       139    MM 
           10       156    MM 
            2       20    b
            3       36    b
            4       52    b
            5       68    b
            6       84    b
            7       100    b
            8       116    b
            9       132    b
            10       148    b
            5       55    alg
            6       83    alg
            7       101    alg
            9       141    alg
            10       159    alg
            }\newdataFigA

\pgfplotstableread{
  x      y        label
             2 0 SMb
             3 0 SMb
             4 0 SMb
            5 1 SMb
            6 2 SMb
            7 1 SMb
            9 1 SMb
            10 1 SMb
}\newdataSMbA
 
\pgfplotstableread{
      x      y        label
            2 0 SMa
            3 0 SMa
            4 0 SMa
            5 1 SMa
            6 4 SMa
            7 7 SMa
            9 10 SMa
            10 11 SMa
 }\newdataSMaA

\pgfplotstableread{
  x      y        label
            2 0 MMa0
             3 0 MMa0
             4 0 MMa0
             5 2 MMa0
             6 5 MMa0
             7 8 MMa0
             8 9 MMa0
             9 11 MMa0
             10 12 MMa0
 }\newdataMMaA

\pgfplotstableread{
            x      y        label
           2       26    MM 
            2       34    b
            3       33    MM 
            3       162    b
            4       40    MM 
            4       290    b
            5       125    MM 
            5       65    alg
            5       418    b
            6       287    MM 
            6       242    alg
            6       546    b
            7       495    MM 
            7       439    alg
            7       674    b
            8       625    MM 
            8       802    b
            9       839    MM 
            9       777    alg
            9       930    b
            10       1003    MM 
            10       935    alg
            10       1058    b
           }\dataFigB

\pgfplotstableread{
  x      y        label
            2 0 MMa0
             3 0 MMa0
             4 0 MMa0
             5 2 MMa0
             6 5 MMa0
             7 8 MMa0
             8 9 MMa0
             9 11 MMa0
             10 12 MMa0
 }\dataMMaB

\pgfplotstableread{
  x      y        label
     2 0 SMa
     3 0 SMa
      4 0 SMa
      5 0 SMa
      6 3 SMa
      7 6 SMa
      9 10 SMa
      10 11 SMa
 }\dataSMaB

\pgfplotstableread{
  x      y        label
     2 0 SMb
     3 0 SMb
      4 0 SMb
      5 2 SMb
      6 13 SMb
      7 24 SMb
      9 1 SMb
      10 1 SMb
 }\dataSMbB

 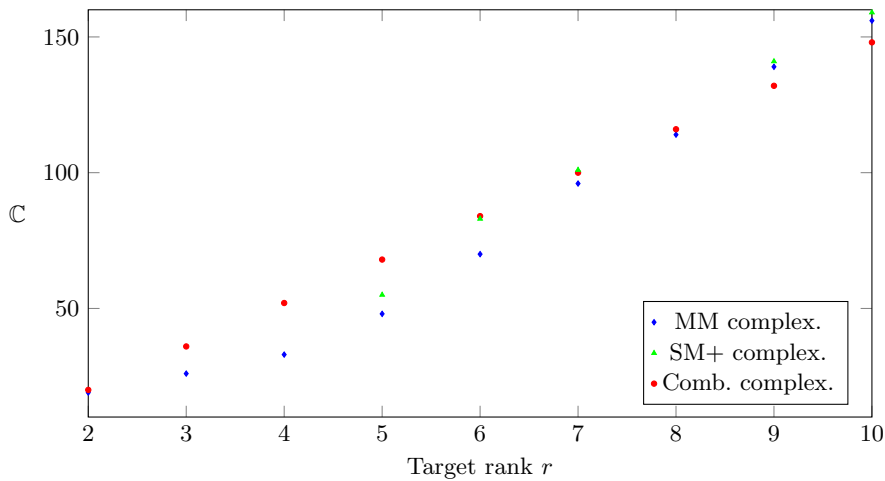
\begin{figure}[ht!]
  \centering
  \begin{tikzpicture} 
    \begin{axis}
      [
      legend entries={MM complex., SM+ complex., Comb. complex.},
      legend pos=south east,
      height=7.0cm,
      width=12.0cm,
      xmin=2,xmax=10,
      ymin=10,ymax=160,
      mark size=1pt,
      xlabel={Target rank $r$},
      ylabel={$\mathbb{C}$},
      ylabel style={rotate=-90}]
        \addplot [
            scatter,
            only marks,
            point meta=explicit symbolic,
            scatter/classes={
                MM={mark=diamond*,blue},
                alg={mark=triangle*,green},
                b={mark=*,red}    
            },
        ] table [meta=label] {\newdataFigA };
    \end{axis}
  \end{tikzpicture}
  \caption{Comparison between the theoretical 
$\log_2$  complexities $\mathbb{C}$ of \MMfq/\SMpfqm (the best one, hybrid and punctured version) and of the combinatorial attack for 
  RD instances with fixed $(m,n,k)=(31,33,15)$ and various values of $r$. The rank Gilbert-Varshamov bound is $d_{\text{RGV}}(m,n,k,q = 2)=10$.
   \label{fig:graph_MM_combi}}
\end{figure}

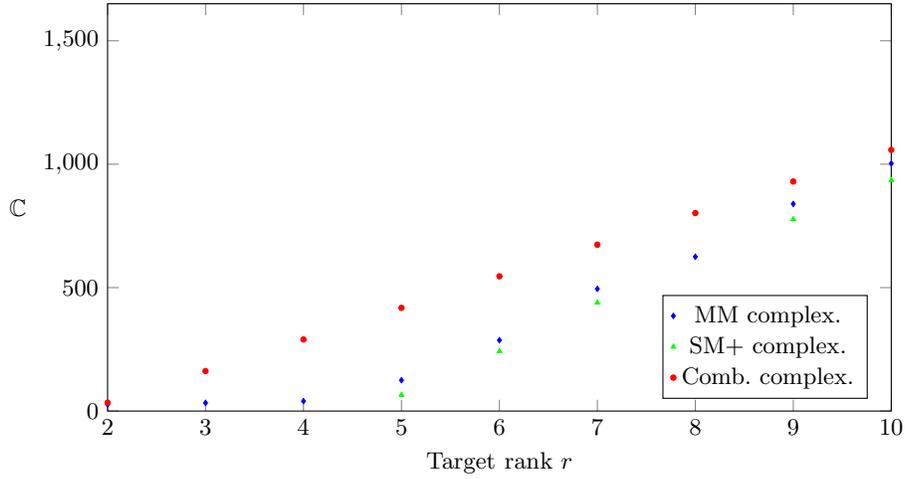
\begin{figure}[ht!]
  \centering
  \begin{tikzpicture} 
    \begin{axis}
      [
      legend entries={MM complex., SM+ complex., Comb. complex.},
      legend pos=south east,
      height=7.0cm,
      width=12.0cm,
      mark size=1pt,
      xmin=2,xmax=10,
      ymin=0,ymax=1650,
      xlabel={Target rank $r$},
      ylabel={$\mathbb {C}$},
      ylabel style={rotate=-90}]
        \addplot [
            scatter,
            only marks, 
            point meta=explicit symbolic,
            scatter/classes={
                MM={mark=diamond*,blue},
                alg={mark=triangle*,green},
                b={mark=*,red}    
            },
        ] table [meta=label] {\dataFigB };
    \end{axis}
  \end{tikzpicture}
  \caption{Same parameters as \cref{fig:graph_MM_combi} but with $q=2^8$.
   \label{fig:graph_MM_combiq}}
\end{figure}

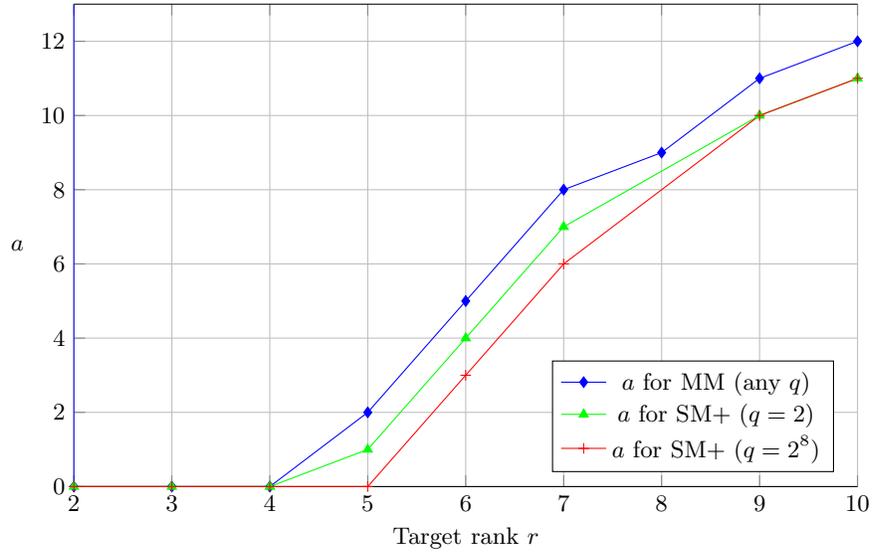
\begin{figure}[ht!]
  \centering
  \begin{tikzpicture} 
    \begin{axis}
      [
      axis y line*=left,
      y axis line style=blue,
      legend entries={$a$ for MM (any $q$),$a$ for SM+ ($q=2$), $a$ for SM+ ($q=2^8$)},
      legend pos= south east,
      height=8.0cm,
      width=12.0cm,
      grid=major,
      xmin=2,xmax=10,
      ymin=0,ymax=13,
      xlabel={Target rank $r$},
      ylabel={$a$},
      ylabel style={rotate=-90}]
        \addplot [blue, mark=diamond*] table [meta=label] {\newdataMMaA  };
        \addplot [green, mark=triangle* ] table [meta=label] {\newdataSMaA};
        \addplot [red,  mark=+] table [meta=label] {\dataSMaB};
      \end{axis}
  \end{tikzpicture}
  \caption{Optimal values of $a$ with $m=31, n=33, k=15$,  $q=2$ or $q=2^8$, for \MMfq and \SMpfqm.\label{fig:optimala}}
\end{figure}

We illustrate in \cref{fig:graph_MM_combiq} the fact that our
approach over $\Fqm$ becomes interesting compared to \MMfq as $q$
increases, and for small values of $r$.  This can be explained by the
fact that \SMpfqm contains two blocks of variables, the $c_T$'s and
the $x_i$'s, and introducing the $x_i$'s variables has a computational
cost that make \MMfq competitive for large $r$. For large $q$, the
cost of the hybrid approach becomes higher and the \SMpfqm{} approach
more competitive, as it can solve with a smaller $a$ at a larger
$b$. We plot in
\cref{fig:optimala} the optimal values of $a$
and compare the \MMfq approach with \SMpfqm for $q=2$ and $q=2^8$.

\noindent 
\textbf{General picture of the complexities of generic RD instances.} 
Even if it is difficult to draw general conclusions for the complexity
of the different attacks against the Rank Decoding problem, our simulations
seem to show that the \MMfq and \SMpfqm algebraic attacks are particularly
more efficient than combinatorial attacks when, roughly, $r$ is small and $m$ is not too small
(typically the case of original LRPC parameters).

Regarding the case of the harder zone typically used in code-based cryptography, namely
$m=n,k=n/2$ and $r$ close to the Rank Gilbert-Varshamov bound,
our results seem to indicate that all attacks, both algebraic and combinatorial, 
have similar complexities.
Surprisingly enough, this seems to remain true even in the case of greater value of $q$ ($q>2$).

\section*{Conclusion}
We have presented here a new algebraic modeling for the RD problem,
\SMpfqm, that takes advantage of the $\fqm$-linearity of the problem
to adapte the Support Minors Modeling \SMfq for MinRank instances to
the RD case. This modeling extends the MaxMinors Modeling \MMfq for
systems that are not overdetermined. We have given a proof for the
number of linearly independent polynomials in \SMfqm, and good
heuristic explanation for the number of linearly independent
polynomials in \SMpfqm.

From the computational point of view, when the field $q$ is small, the
\MMfq Modeling is faster to solve than the combinatorial approach,
whereas it is the opposite for $r$ close to the rank GV
bound. However, when $q$ increases, the algebraic approaches \MMfq and
\SMpfqm becomes faster, and for small values of $r$ the \SMpfqm
Modeling beats the \MMfq Modeling.

Finally, we have proposed an hybrid approach that reduces the solving
of a MinRank (resp. RD) instance to the solving of several smaller
instances. This has the advantage to apply to any solving algorithm
for MinRank (resp. RD).

\subsubsection*{Acknowledgements.}
The authors thank the reviewers for their careful reading of the paper
and their helpful comments.

This research was funded by the French {\it Agence Nationale de la Recherche} and
 {\it plan France 2030} program under grant ANR-22-PETQ-0008 PQ-TLS.

\appendix
\section{Missing proofs from Section \ref{sec:two}}
\label{sec:proofs}

It will be helpful to notice that  \cref{lemma:generality} implies
\begin{lemma}\label{lem:minors}
Let $T\subset\Iint{1}{n-k-1}$, then 
\begin{eqnarray}
\minor{\Hm_{\yv}}{T,T+k+1}&=&1 \\
\minor{\Hm_{\yv}}{T,I}&=&0 \;\;\text{if
  $I \cap \Iint{k+2}{n} \nsubseteq T+k+1$}
\end{eqnarray}
\end{lemma}
\begin{proof}
This follows immediately from the fact that $\Hm_{\yv}$ is in systematic form in its $n-k-1$ last coordinates (i.e. for the positions $j \in \Iint{k+2}{n}$): $\Hm_y = \begin{pmatrix} * & \Im_{n-k-1} \end{pmatrix}$. Indeed 
$\minor{\Hm_{\yv}}{T,T+k+1}=|\Im_s|=1$ where $s = \#T$. The other minor is 0 since it contains a column which is $0$. \qed
\end{proof}
\subsection{Proof of Proposition \ref{prop:Q0}}
\label{proof:Q0}
Let us recall this proposition.
\repeatproposition{prop:Q0}

\begin{proof}
 We first observe that $Q$ in $\mathcal Q_0$ is of the form$Q_{T+k+1}$ with
$T\subset\Iint{1}{n-k-1},~\# T = r+1$. By definition we have
$(\xv\Gm+\yv)\trsp{\Hm_{\yv}}=0$ and hence by using the Cauchy-Binet formula \eqref{eq:Cauchy-Binet} we obtain
  \begin{align*}
    0&=  \minor{    \begin{pmatrix}
        \xv\Gm+\yv\\\Cm
      \end{pmatrix}
    \trsp{\Hm_{\yv}}}{\any,T} = \sum_{\substack{I\subset\Iint{1}{n}\\\#I=r+1}}  \minor{\Hm_{\yv}}{T,I}Q_I.
  \end{align*}
  We then use Lemma \ref{lem:minors}: $\minor{\Hm_{\yv}}{T,T+k+1}=1$, and $\minor{\Hm_{\yv}}{T,I}=0$ if
  $I\subset\Iint{k+2}{n}$, $I\ne T+k+1$.
  The previous equation expresses $Q_{T+k+1} \in \mathcal Q_0$ in
  terms of the $Q_I$'s in $Q_{\ge 1}$, namely
    $\textstyle{Q_{T+k+1} = - \sum_{Q_I\in \mathcal Q_{\ge 1}} \minor{\Hm_{\yv}}{T,I}Q_I}$.
    \qed
\end{proof}

\subsection{Proofs of Propositions \ref{prop:LTQI} and  \ref{prop:Q1} }
\label{proof:LTQIQ1}
For the proofs of~\cref{prop:Q1,prop:LTQI}, we recall that we use the grevlex
monomial ordering on the variables $x_1>\dots>x_k>c_T$ with the $c_T$'s ordered according to a reverse lexicographical ordering on $T$:
$c_{T'}>c_{T}$ if $t'_j=t_j$ for $j<j_0$ and $t'_{j_0}>t_{j_0}$ where
$T=\lbrace t_1<\dots<t_r\rbrace$ and
$ T'=\lbrace t'_1<\dots<t'_r\rbrace$. We denote by $\LT(f)$ the
leading term of a polynomial $f$ with respect to this term order.

We will also make use of the following lemma.
\begin{lemma}\label{lem:Q1}
Let $Q_I$ be an equation in $\mathcal Q_{\geq 2}$. We have
\begin{eqnarray*}
\LT(Q_I) &= &x_{i_1}c_{I_1} \\
 Q_I &= & x_{i_1}c_{I_1} \underbrace{-\xv\Gm_{\any,i_2}c_{I_2} + \dots + (-1)^{r}\xv\Gm_{\any,i_{r+1}}c_{I_{r+1}}}_{ \text{smaller terms of degree 2}}\\
      &   & \underbrace{- y_{i_2}c_{I_2} + \dots + (-1)^{r} y_{i_{r+1}}c_{I_{r+1}}}_{ \text{smaller terms of degree 1}}
\end{eqnarray*}
where $I=\lbrace i_1 < \dots < i_{r+1}\rbrace$ and $I_1 \eqdef I \setminus \{i_1\}$. The leading terms of such $Q_I$'s are all different and the variables $\lbrace c_{J+k+1} \rbrace_{J\subset\Iint{1}{n-k-1}}$ do not
appear in $Q_I$.
\end{lemma}
\begin{proof}
Since $Q_I$ is in $\mathcal Q_{\geq 2}$ we know that $i_1 \le k$.
We have
\begin{align*}
  Q_I &= \minor{
        \begin{pmatrix}
          \xv\Gm+\yv\\\Cm
        \end{pmatrix}
  }{\any,I} = \sum_{i_u\in I}(-1)^{1+u}(\xv\Gm_{\any,i_u} + y_{i_u})c_{I\setminus\lbrace i_u\rbrace}.
\end{align*}
Taking $\Gm$ and $\yv$ as in~\cref{lemma:generality}, for any
$i_u\in I^-=I\cap\Iint{1}{k}$ (and at least $i_1\in I^-$ by
assumption), we have $\xv\Gm_{\any,i_u}+\yv_{i_u}=x_{i_u}$. Let
$I_u = I\setminus\lbrace i_u\rbrace$ for $1\le u \le r+1$, then for
the chosen ordering we have $I_1>I_2>\dots>I_{r+1}$. The ordered terms
in $Q_I$ are then
\begin{align*}
   Q_I = & x_{i_1}c_{I_1} \underbrace{-\xv\Gm_{\any,i_2}c_{I_2} + \dots + (-1)^{r}\xv\Gm_{\any,i_{r+1}}c_{I_{r+1}}}_{ \text{smaller terms of degree 2}}\\
         & \underbrace{- y_{i_2}c_{I_2} + \dots + (-1)^{r} y_{i_{r+1}}c_{I_{r+1}}}_{ \text{smaller terms of degree 1}}
\end{align*}
so that$\LT(Q_I) = x_{i_1}c_{I_1}$ and these leading terms are different for
all the equations. For the last point, we observe that $\lbrace i_1<i_2\rbrace\subset \Iint{1}{k+1}$. This implies that
for any $i_u\in I$, the set $I\setminus\lbrace i_u\rbrace$ contains at
least one of $i_1,i_2$ so that it is not included in
$\Iint{k+2}{n}$, from which it follows that the variables $\lbrace c_{J+k+1} \rbrace_{J\subset\Iint{1}{n-k-1}}$ do not
appear in $Q_I$. \qed
\end{proof}
 
 We are ready now to prove Proposition \ref{prop:LTQI}:
 \repeatproposition{prop:LTQI}

\begin{proof}
Lemma \ref{lem:minors} already proves that the equations in $\mathcal Q_{\ge 2}$ are linearly independent. Consider now a $P_J \in \cP$.
Here
$J\subset\Iint{1}{n-k-1},~\#J =r$. By using the special
shape of $\Hm_{\yv}$ we have
\begin{align*}
  P_J &= \minor{\Cm\trsp{\Hm_{\yv}}}{\any,J} = \sum_{\substack{T\subset\Iint{1}{n}\\\#T=r}} c_T\minor{\Hm_{\yv}}{J,T} = \sum_{\substack{T\subset\Iint{1}{n},\#T=r,\\ T\cap\Iint{k+2}{n} \subset J+k+1}} c_T\minor{\Hm_{\yv}}{J,T}\\
      &= c_{J+k+1} + \sum_{\substack{T\subset\Iint{1}{n},\#T=r,\\ T\cap\Iint{k+2}{n} \subset J+k+1, T\cap\Iint{1}{k+1} \ne \emptyset}} c_T\minor{\Hm_{\yv}}{J,T}
\end{align*}
We used here again Lemma \ref{lem:minors}.
Note that the $c_T$'s in the sum are all smaller than $c_{J+k+1}$, so
that $c_{J+k+1}$ is the leading term of $P_J$ and does not appear in any other
$P_{J'}$. This shows that the polynomials in
$\mathcal P \cup \mathcal Q_{\ge 2}$ are linearly independent, as they
have distinct leading terms, and concludes the proof
of~\cref{prop:LTQI}. 
 \qed
\end{proof}

Let us now recall  Proposition \ref{prop:Q1} before proving it.
\repeatproposition{prop:Q1}

\begin{proof}
Consider $  \minor{\begin{pmatrix}
      \xv\Gm+\yv\\\Cm
    \end{pmatrix}
  \trsp{\Hm}}{\any,J\cup\lbrace n-k\rbrace}$. On one hand, we have with the Cauchy-Binet formula
    \begin{equation}\label{eq:xGpymH}
    \minor{\begin{pmatrix}
      \xv\Gm+\yv\\\Cm
    \end{pmatrix}
  \trsp{\Hm}}{\any,J\cup\lbrace n-k\rbrace}
  =  \sum_{\substack{I\subset\Iint{1}{n},\#I=r+1}} \minor{\Hm}{J\cup\lbrace n-k\rbrace,I} Q_I .
    \end{equation}
On the other hand, we use the particular shapes for $\Hm$, $\yv$ and $\hv$ given in~\cref{lemma:Hsystematic}:
\begin{eqnarray*}
\Hm & =& 
  \begin{pmatrix}
    \Hm_{\yv}\\\hv
  \end{pmatrix} \\
  \yv & = & \begin{pmatrix} \zerom_k & 1 & \any \end{pmatrix}\\
  \hv & = & \begin{pmatrix}
    \any & 1 & \zerom_{n-k-1}
  \end{pmatrix}
\end{eqnarray*}   and obtain
\begin{align*}
  \begin{pmatrix}
    \xv\Gm+\yv\\\Cm
  \end{pmatrix}
  \trsp{\Hm} =
  \begin{pmatrix}
    \yv\trsp{\Hm}\\\Cm\trsp{\Hm}
  \end{pmatrix} =
  \begin{pmatrix}
    \yv\trsp{\Hm_{\yv}}& \yv \trsp{\hv} \\ \Cm\trsp{\Hm_{\yv}} & \Cm\trsp{\hv}
  \end{pmatrix}
  =
  \begin{pmatrix}
    \zerov_{n-k-1} & 1\\
    \Cm\trsp{\Hm_{\yv}} & \Cm\trsp{\hv}
  \end{pmatrix},
\end{align*}
so that for any $J\subset\Iint{1}{n-k-1},~\#J =r$ we get
\begin{align*}
  \minor{\begin{pmatrix}
      \xv\Gm+\yv\\\Cm
    \end{pmatrix}
  \trsp{\Hm}}{\any,J\cup\lbrace n-k\rbrace}
  &=  \minor{  \begin{pmatrix}
      \zerov & 1\\
      \Cm\trsp{\Hm_{\yv}} & \Cm\trsp{\hv}
    \end{pmatrix}}{\any,J\cup\lbrace n-k\rbrace}
 &= (-1)^r \minor{\Cm\trsp{\Hm_{\yv}}}{\any,J} = (-1)^r P_J.
\end{align*}
By using
\begin{align*}
  \minor{\Hm}{J\cup\lbrace n-k\rbrace, I}
  &=
    \begin{cases}
      0 & \text{ if } I\cap\Iint{k+2}{n}\not\subset J+k+1\\
      (-1)^r & \text{ if } I = \lbrace k+1 \rbrace \cup (J+k+1),
    \end{cases}
\end{align*}
we then have
$P_J = \underbrace{Q_{\lbrace k+1\rbrace \cup (J+k+1)}}_{\in \mathcal
  Q_1} + (-1)^r\sum_{Q_I \in \mathcal Q_{\ge 2}}
\minor{\Hm}{J\cup\lbrace n-k\rbrace,I}Q_I $.

This gives a one-to-one correspondence between equations $P_J$ and
equations $Q_{\lbrace k+1\rbrace\cup J+k+1} \in \mathcal Q_1$. It
remains to show that the
$Q_{\lbrace i_1\rbrace \cup J+k+1}\in\mathcal Q_1$ with $i_1\le k$
reduce to $x_{i_1}P_{J}$ modulo $\mathcal Q_{\ge 2}$.

If $\gv_{i_1} \eqdef \Gm_{\lbrace {i_1} \rbrace, \any}$, we consider $\Hm_{{i_1}}$ a
parity-check matrix of the code
$\mathcal{C}_{i_1} \eqdef \langle {\yv},{\gv_1}, \dots, {\gv_{i_1-1}},
{\gv_{i_1+1}}, \dots, {\gv_{k}} \rangle$ such that
$\trsp{\Hm_{i_1}} = \begin{pmatrix} \trsp{\Hm_{\yv}} &
  \trsp{{\ev_{i_1}}} \end{pmatrix}$ and where $\ev_{i_1}$ is the ${i_1}$-th
canonical basis vector in $\ff{q}^n$.  Since $ \gv_{i_1}\trsp{\ev_{i_1}} = 1$,
we have
\begin{equation*}
  \begin{pmatrix}
    \xv\Gm+\yv\\\Cm
  \end{pmatrix}
  \trsp{\Hm_{i_1}} = 
  \begin{pmatrix}
    x_{i_1}\gv_{i_1} \trsp{\Hm_{i_1}} \\
    \Cm\trsp{\Hm_{{i_1}}} 
  \end{pmatrix} = \begin{pmatrix}
    \zerom & x_{i_1} \\
    \Cm\trsp{\Hm_{\yv}} & \Cm\trsp{{\ev_{i_1}}} 
  \end{pmatrix}.
\end{equation*}
For $J \subset \Iint{1}{n-k-1},~\#J = r$, one obtains
\begin{align*}
  \minor{\begin{pmatrix}
      \xv\Gm+\yv\\\Cm
    \end{pmatrix}\trsp{\Hm_{i_1}}}{\any,J\cup\lbrace n-k\rbrace}
& =\minor{  \begin{pmatrix} \zerov & x_{i_1}\\
      \Cm\trsp{\Hm_{\yv}} & \Cm\trsp{\hv_{i_1}}
    \end{pmatrix}}{\any,J\cup\lbrace n-k\rbrace}\\
   = \sum_{\substack{I \subset \Iint{1}{n}, \#I = r+1}}\minor{\Hm_{i_1}}{J\cup\lbrace n-k\rbrace,I}Q_I    =& (-1)^{r}x_{i_1} \minor{\Cm\trsp{\Hm_{\yv}}}{\any,J}= (-1)^{r}x_{i_1} P_J.
\end{align*}
By Laplace expansion along the last row, we have
$\minor{\Hm_{i_1}}{J\cup\lbrace n-k\rbrace, I}=0$ if ${i_1}\notin I$
and
$\minor{\Hm_{i_1}}{J\cup\lbrace n-k\rbrace,
  I}=(-1)^{r+1+\Pos({i_1},I)}\minor{\Hm_{\yv}}{J, I\setminus\lbrace
  {i_1}\rbrace}$ if ${i_1}\in I$, where $\Pos({i_1},I)$ denotes the position of $i_1$ in the ordered set $I$ (1 if it is the first element).
 We deduce from this that
\begin{align*}
  x_{i_1} P_J &= \sum_{\substack{I \subset \Iint{1}{n}, \#I = r+1, {i_1}\in I}}  (-1)^{1+\Pos({i_1},I)}\minor{\Hm_{\yv}}{J, I\setminus\lbrace {i_1}\rbrace}Q_I \\
               &= Q_{\lbrace {i_1}\rbrace \cup (J + k +1)} + \sum_{Q_I\in\mathcal Q_{\ge 2}, {i_1}\in I}(-1)^{1+\Pos({i_1},I)}\minor{\Hm_{\yv}}{J, I\setminus\lbrace {i_1}\rbrace}Q_I .
\end{align*}
Note that by the previous results,
$\LT(Q_{\lbrace i_1\rbrace+J+k+1})=x_{i_1}c_{J+k+1}$ so that all
equations in
$\textstyle{\mathcal P\cup \bigcup_{j=1}^k x_j \mathcal P\cup\mathcal Q_{\ge 2}}$
are linearly independent.
\end{proof}

\subsection{Proof of Proposition \ref{prop:Nb}}
\label{proof:Nb}
Let us first recall this proposition.
\repeatproposition{prop:Nb}

\begin{proof}
The set $\mathcal B_b$ clearly contains linearly independent
equations, since their leading terms are all different:
\begin{align*}
  \LT({x_{i_1}}^{\alpha_{i_1}}\dots {x_{k}}^{\alpha_k}Q_I) = {x_{i_1}}^{\alpha_{i_1}+1}\dots {x_{k}}^{\alpha_k}c_{I\setminus\lbrace i_1\rbrace}.
\end{align*}
The number of polynomials in $\mathcal B_b$ is the number of sets $I$ and $(\alpha_{i_1},\dots,\alpha_k)$:
\begin{align*}
  \Nbfqm{} &= \sum_{i_1=1}^k \sum_{i_2=i_1+1}^{k+1}\binom{n-i_2}{r-1}\binom{k-i_1+1+b-2}{b-1}
\end{align*}
which gives~\cref{eq:Nbfqm}, considering the identities $\textstyle{\sum_{i_2=i_1+1}^{k+1}\binom{n-i_2}{r-1}=\binom{n-i_1}{r}-\binom{n-k-1}{r}}$ and $\textstyle{\sum_{i_1=1}^k \binom{k-i_1+1+b-2}{b-1} = \binom{k+b-1}{b}}$. 
The number of monomials comes from the fact that the variables $c_{J+k+1}$ do not appear in $\mathcal Q_{\ge 2}$. The inequality $\mathcal N_b < \mathcal M_b-1$ is easy to derive using previous identities and $\textstyle{\binom{n-i_1}{r}<\binom{n-1}{r}}$ for all $i_1\ge 1$.

We will now show that the polynomials $x_{j}Q_I$ for $1\le j < i_1$,
$Q_I\in\mathcal Q_{\ge 2}$ reduce to zero modulo $\mathcal B_2$, which
is sufficient to conclude the proof. The number of such polynomials is
equal to the number of sets
$K=\lbrace k_1 < k_2 < \dots < k_{r+2}\rbrace \subset\Iint{1}{n}$ such that $k_3\le k+1$, and we are going to construct the
same number of independent syzygies between the polynomials at
bi-degree $(2,1)$. Indeed, for any such $K$, we have the relation
\begin{align}
  \minor{
  \begin{pmatrix}
    \xv\Gm + \yv\\\xv\Gm + \yv\\\Cm
  \end{pmatrix}}{\any,K} &= 0.
\end{align}
By Laplace expansion along the first row, one obtains
\begin{align*}
0    & = x_{k_1}Q_{\lbrace k_2,\dots,k_{r+2}\rbrace} - \sum_{u=2}^{r+2}  (-1)^{u}\left(\sum_{j=1}^k x_j \Gm_{j,k_u} + y_{k_u} \right) Q_{K\setminus\{k_u\}}.
\end{align*}
Since $\vert K \cap \Iint{1}{k+1} \vert \geq 3$, we obtain syzygies
between the relevant $Q_I$, say those such that
$\vert I \cap \Iint{1}{k+1} \vert \geq 2$. We will now show that those
syzygies are linearly independent. To this end, we order the $Q_I$'s
according to a grevlex order on the subsets $I$ as for the $c_T$
variables. The largest $Q_I$ is $Q_{\Iint{n-r}{n}}$, the
smallest one is $Q_{\Iint{1}{r+1}}$.  The syzygy associated to
$K$ is given by
\begin{equation*}
  {\mathcal G}^K \eqdef \begin{pmatrix}
    \underbrace{0}_{I\not \subset K}, \underbrace{(-1)^{1+u}\sum_{j=1}^k x_j\Gm_{j,k_u} + \yv_{k_u}}_{K\setminus I = \{k_u\}}
  \end{pmatrix}_{I\subset\Iint{1}{n}, \#I=r+1}.
\end{equation*}
The largest set $I$ such that the coefficient in front of $Q_I$ in
${\mathcal G}^K$ is non-zero is
$I = K_1 = K \setminus \lbrace k_1 \rbrace$ and this coefficient is
$x_{k_1}$.  The syzygies which have the same leading position
$Q_{K_1}$ as ${\mathcal G}^K$ are the
${\mathcal G}^{K_1 \cup \lbrace j \rbrace}$ for $1 \leq j <
k_1$. Finally, the highest degree part in the coefficient in front of
$Q_{K_1}$ in ${\mathcal G}^{K_1 \cup \lbrace j \rbrace}$ is $x_{j}$,
which shows that all the ${\mathcal G}^{K_1 \cup \lbrace j \rbrace}$
are linearly independent for $1 \leq j \leq k_1$.
\qed
\end{proof}

\subsection{Proof of~\cref{prop:QIofPiJ}}
\label{proof:QIofPiJ}
Let us recall this proposition.
\repeatproposition{prop:QIofPiJ}
\begin{proof}
  For this
purpose, we introduce the $\ell$-th Frobenius iterate of the $P_J$'s,
that has the advantage to satisfy the relation $\minor{\Mm^{[\ell]}}{} =\minor{\Mm}{}^{[\ell]} $ for
any square matrix $\Mm$. This is equivalent to using the unfolded equations
$P_{i,J}$ thanks to the following relation:
for any $J\subset\Iint{1}{n-k-1},\#J=r$ we have
  \begin{align*}
    \langle P_{i,J} : 1\le i \le m \rangle_{\fqm}
    &=
      \langle P_{J}^{[\ell]} : 0\le \ell \le m-1\rangle_{\fqm}.
  \end{align*}
Indeed,
  $\textstyle{P_{i,J} = \tr(\beta_i^\star P_J) = \sum_{\ell=0}^{m-1}
    (\beta_i^{\star})^{[\ell]}P_J^{[\ell]}}$ and
 $ \textstyle{P_J^{[\ell]} = \sum_{i=1}^m \beta_i^{[\ell]} P_{i,J}}$.

  For fixed $ 0\le \ell \le m-1$ and
  $T\subset\Iint{1}{n-k-1}, \#T=r+1$, we consider the minor
  \begin{align*}
    \Gamma_{\ell,T} &\eqdef \minor{
    \begin{pmatrix}
      \xv\Gm + \yv\\\Cm
    \end{pmatrix}
    \trsp{(\pow{\Hm_{\yv}}{\ell})}}{\any,T}.
  \end{align*}
  By Laplace expansion along the first row, this minor can be viewed as a combination with coefficients in $\ff{q^m}[x_i]$ between maximal minors of $\Cm\trsp{(\pow{\Hm_{\yv}}{\ell})}_{\any,T}$, and these minors are exactly the $\pow{P_J}{\ell}$'s for $J\subset T$. The normal form of $\Gamma_{\ell,T}$ with respect to $\langle P_{i,J}\rangle=\langle P_{J}^{[\ell]}\rangle$ is then 0. Also, using the Cauchy-Binet formula, each minor is a linear combination of the $Q_I$'s, given by   
  \begin{align*}
    \tilde{Q}_{T+k+1} + \sum_{\substack{I\subset\Iint{1}{n}, \#I=r+1,\\I\cap\Iint{ k+1}{n}\subsetneq T+k+1}} \tilde{Q}_{I}\minor{\pow{\Hm_{\yv}}{\ell}}{T,I} = 0. 
  \end{align*}
  To conclude the proof, we use the fact that the set of previous equations for all $0\le \ell\le m-1$ generate the same vector space over $\ff{q^m}$ as the set of equations
  \begin{align*}
    \tr(\beta_i^\star)\tilde{Q}_{T+k+1} + \sum_{\substack{I\subset\Iint{1}{n}\\\#I = r+1\\I\cap\Iint{k+1}{n}\subsetneq T+k+1}}\tr(\beta_i^\star \minor{\Hm_{\yv}}{T,I})\tilde{Q}_I =0,
  \end{align*}
  for all $1\le i \le m$.
\qed
\end{proof}

\end{document}